\documentclass[12pt]{article}

\usepackage{tikz}
\usetikzlibrary{shapes}
\usepackage{amsmath,amsthm,amsfonts,amssymb}
\usepackage{dsfont}
\usepackage{hyperref}
\usepackage{graphicx}
\usepackage{xcolor}
\usepackage{bm}
\usepackage{algorithm}
\usepackage{algpseudocode}
\usepackage{enumitem}
\usepackage{subcaption}
\usepackage{booktabs}

\usepackage{geometry}
\usepackage[noblocks]{authblk}

\usepackage[authoryear]{natbib} 
\title{Anticipatory Fictitious Play}

\author[12]{Alex Cloud\footnote{Affiliations refer to where this work was carried out and are no longer current.}}
\author[1]{Albert Wang\footnote{Corresponding author, email: alwang@alumni.stanford.edu.}}
\author[1]{Wesley Kerr}
\affil[1]{Riot Games Technology Research, Los Angeles, CA}
\affil[2]{North Carolina State University, Raleigh, NC}

\date{December 2, 2022}

\theoremstyle{definition}
\newtheorem{thm}{Theorem}
\newtheorem{prop}{Proposition}
\newtheorem*{remark*}{Remark}

\newtheorem{lemma}{Lemma}
\newtheorem*{lemma*}{Lemma}
\newtheorem{defn}{Definition}

\newcommand{\br}{\texttt{BR}}

\newcommand{\xbar}{\overline{x}}
\newcommand{\ybar}{\overline{y}}
\newcommand{\Proj}{\textnormal{Proj}}

\newcommand{\FP}{\textnormal{FP}}
\newcommand{\AFP}{\textnormal{AFP}}

\newcommand{\R}{\mathbb{R}} 
\newcommand{\N}{\mathbb{N}} 
\newcommand{\T}{\intercal} 

\newcommand{\ind}{\mathds{1}} 

\newcommand{\Prob}{\mathcal{P}} 
\newcommand{\A}{\mathcal{A}} 
\newcommand{\St}{\mathcal{S}} 
\newcommand{\Obs}{\mathcal{O}} 
\newcommand{\X}{\mathcal{X}}
\newcommand{\Hist}{\mathcal{H}} 

\begin{document}



\maketitle

\begin{abstract}
Fictitious play is an algorithm for computing Nash equilibria of matrix games. Recently, machine learning variants of fictitious play have been successfully applied to complicated real-world games. This paper presents a simple modification of fictitious play which is a strict improvement over the original: it has the same theoretical worst-case convergence rate, is equally applicable in a machine learning context, and enjoys superior empirical performance. We conduct an extensive comparison of our algorithm with fictitious play, proving an optimal $O(t^{-1})$ convergence rate for certain classes of games, demonstrating superior performance numerically across a variety of games, and concluding with experiments that extend these algorithms to the setting of deep multiagent reinforcement learning. 
\end{abstract}


\section{Introduction}
Matrix games (also known as {\em normal-form games}) are an abstract model for interactions between multiple decision makers. Fictitious play (FP) (\citet{brown1951iterative}) is a simple algorithm for two-player matrix games. In FP, each player starts by playing an arbitrary strategy, then proceeds iteratively by playing the best strategy against the empirical average of what the other has played so far. In some cases, such as two-player, zero-sum games, the empirical average strategies will converge to a Nash equilibrium.

Although there are more efficient algorithms for computing Nash equilibria in matrix games \citep{adler2013equivalence, shoham2008multiagent}, there are a few reasons why fictitious play remains a topic of interest. First, it serves as a model for how humans might arrive at Nash equilibria in real-world interactions \citep{luce1989games, brown1951iterative, conlisk1993adaptation}. Second, FP is extensible to real-world games which are large and complicated. Our work is primarily motivated by the secondary application.

\sloppy The initial step towards extending FP to real-world games was by \citet{kuhn1953extensive}, which established the equivalence of normal-form games (represented by matrices) and extensive-form games (represented by trees with additional structure). Loosely speaking, this means that results which apply for matrix games may also apply to much more complicated decision making problems, such as ones that that incorporate temporal elements or varying amounts of hidden information.
Leveraging this equivalence, \citet{heinrich2015fictitious} proposed an extension of FP to the extensive-form setting, full-width extensive-form fictitious play (XFP), and proved that it converges to a Nash equilibrium in two-player, zero-sum games. \citet{heinrich2015fictitious} also proposed Fictitious Self Play (FSP), a machine learning approximation to XFP. In contrast to XFP, which is intractable for real-world games whose states cannot be enumerated in practice, FSP relies only on basic operations which can be approximated in a machine learning setting, like averaging (via supervised learning) and computing best responses (via reinforcement learning). In this way, FSP provides a version of fictitious play suitable for arbitrarily complex two-player, zero-sum games. Not long after the introduction of FSP, \citet{lanctot2017unified} presented Policy Space Response Oracles (PSRO), a general framework for fictitious-play-like reinforcement learning algorithms in two-player, zero-sum games. These ideas were employed as part of the groundbreaking AlphaStar system that defeated professional players at StarCraft II  \citep{vinyals2019grandmaster}.

We introduce anticipatory fictitious play (AFP), a simple variant of fictitious play which is also reinforcement-learning-friendly. In contrast to FP, where players iteratively update to exploit an estimate of the opponent’s strategy, players in AFP update proactively to respond to the strategy that the opponent would use to exploit them.

We prove that AFP is guaranteed to converge to a Nash equilibrium in two-player zero-sum games and establish an optimal convergence rate for two classes of games that are of particular interest in learning for real world games \citep{balduzzi2019open}, a class of ``cyclic'' games and a class of ``transitive'' games. Numerical comparisons suggest that in AFP eventually outperforms FP on virtually any game, and that its improvement over FP improves as games get larger. Finally, we propose a reinforcement learning version of AFP that is implementable as a one-line modification of an RL implementation of FP, such as FSP. These algorithms are applied to Running With Scissors \citep{vezhnevets2020options}, a stochastic, competitive multiagent environment with cyclic dynamics. 

\subsection{Related work}
Aside from the literature on fictitious play and its extension to reinforcement learning, there has been substantial work on ``opponent-aware'' learning algorithms. These algorithms incorporate information about opponent updates and are quite similar to anticipatory fictitious play.

In the context of evolutionary game theory, \citet{conlisk1993adaptation} proposed an ``extrapolation process,'' whereby two players in a repeated game each forecast their opponents' strategies and respond to those forecasts. Unlike AFP, where opponent responses are explicitly calculated, the forecasts are made by linear extrapolation based on the change in the opponent's strategy over the last two timesteps. \citet{conlisk1993adaptive} proposed two types of ``defensive adaptation,'' which are quite similar in spirit to AFP but differ in some important details; most importantly, while they consider the opponent's empirical payoffs at each step, they do not respond directly to what the opponent is likely to play given those payoffs.

\citet{shamma2005dynamic} proposed derivative action fictitious play, a variant of fictitious play in the continous time setting in which a best response to a forecasted strategy is played, like in \citep{conlisk1993adaptation}. The algorithm uses a derivative-based forecast that is analogous to the discrete-time anticipated response of AFP. However, their convergence results rely on a fixed, positive entropy bonus that incentivizes players to play more randomly, and they do not consider the discrete-time case.

\citet{zhang2010multi} proposed Infinitesimal Gradient Ascent with Policy Prediction, in which two policy gradient learning algorithms continuously train against a forecast of the other's policy. Their algorithm represents the core idea of AFP, albeit implemented in a different setting. However, their proof of convergence is limited to $2x2$ games. \citet{foerster2018learning} and \citet{letcher2018stable} take this idea further, modifying the objective of a reinforcement learning agent so that it accounts for how changes in the agent will change the anticipated {\em learning} of the other agents. This line of research is oriented more towards equilibrium finding in general-sum games (e.g. social dilemmas), and less on efficient estimation of equilibria in strictly competitive two-player environments.


\section{Preliminaries}

A  (finite) {\em two-player zero-sum game} (2p0s game) is represented by a matrix $A \in \R^{m \times n}$, so that when player 1 plays $i$ and player 2 plays $j$, the players observe payoffs $(A_{i,j},-A_{i,j})$ respectively. Let $\Delta^k \subset \R^k$ be the set of probability vectors representing distributions over $\{1,\dots,k\}$ elements. Then a {\em strategy} for player 1 is an element $x \in \Delta^m$ and similarly, a strategy for player 2 is an element $y \in \Delta^n$. 

A {\em Nash equilibrium} in a 2p0s game $A$ is a pair of strategies $(x^*, y^*)$ such that each strategy is optimal against the other, i.e., 
\begin{align*}
    x^* \in \underset{x \in \Delta^m}{\arg\max} \; x^\T A y^* \quad \text{and} \quad
    y^* \in \underset{y \in \Delta^n}{\arg\min} \; (x^*)^\T A y.
\end{align*}
The Nash equilibrium represents a pair of strategies that are ``stable'' in the sense that no player can earn a higher payoff by changing their strategy. At least one Nash equilibrium is guaranteed to exist in any finite game \citep{nash1950equilibrium}. 

Nash equilibria in 2p0s games enjoy a nice property not shared by Nash equilibria in general: in 2p0s games, if $(x_1,y_1)$ and $(x_2,y_2)$ are Nash equilibria, then $(x_2,y_1)$ is a Nash equilibrium. In a 2p0s game, we define a Nash {\em strategy} to be be one that occurs as part of some Nash equilibrium. Note that the aforementioned property does not hold in general, so normally it is only valid to describe collections of strategies (one per player) as equilibria.

A {\em solution} to a 2p0s game $A$ is a pair of strategies $(x^*, y^*)$ such that
\begin{align*}
    \min_{y \in \Delta^n} (x^*)^\T A y \leq (x^*)^\T A y^* \leq \max_{x \in \Delta^m} x^\T A y^*.
\end{align*}
We say $v^* = (x^*)^\T A y^*$, which is unique, the {\em value} of the game. Nash equilibria are equivalent to solutions of 2p0s games \citep{shoham2008multiagent}, which is why we use the same notation. Finally, the {\em exploitability} of a strategy is the difference between the value of the game and the worst-case payoff of that strategy. So the exploitability of $x \in \Delta^m$ is $v^* - \min x^\T A$, and the exploitability of $y \in \Delta^n$ is $\max Ay - v^*$.

\subsection{Fictitious play}
Let $e_1, e_2, \dots$ denote the standard basis vectors in $\R^m$ or $\R^n$. Let $\br^k_A$ be the best response operator for player $k$, so that
\begin{align*}
    (\forall y \in \Delta^n) \;\;\; &\br^1_A(y) = \{e_i \in \R^m : i \in \arg\max A y\}; \\
    (\forall x \in \Delta^m) \;\;\; &\br^2_A(x) = \{e_j \in \R^n : j \in \arg\min x^\T A \}.
\end{align*}
Fictitious play is given by the following process. Let $x_1=\xbar_1=e_i$ and $y_1=\ybar_1=e_j$ be initial strategies for some $i$, $j$. For each $t \in \N$, let
\begin{align*}
    x_{t+1} &\in \br^1_A(\ybar_{t});  &y_{t+1} &\in \br^2_A(\xbar_{t}); \\
    \xbar_{t+1} &= \frac{1}{t+1} \sum_{k=1}^{t+1} x_t; &\ybar_{t+1} &= \frac{1}{t+1} \sum_{k=1}^{t+1} y_t.
\end{align*}
In other words, at each timestep $t$, each player calculates the strategy that is the best response to their opponent's average strategy so far. \citet{robinson1951iterative} proved that the pair of average strategies $(\xbar_t, \ybar_t)$ converges to a solution of the game by showing that the exploitability of both strategies converge to zero. 

\begin{thm} \label{thm:FP_converges} (Robinson, 1951) If $\{(x_t,y_t)\}_{t \in \N}$ is a FP process for a 2p0s game with payoff matrix $A \in \R^{m \times n}$, then
\begin{align*}
\lim_{t \rightarrow \infty} \min \xbar_t^\T A = \lim_{t \rightarrow \infty} \max A \ybar_t = v^*, 
\end{align*}
where $v^*$ is the value of the game. Furthermore, a bound on the rate of convergence is given by
\begin{align*}
    \max A \ybar_t - \min \xbar_t^\T A &= O(t^{-1/(m+n-2)}) \text{ for all } t \in \N,
\end{align*}
where $a = \max_{i,j} A_{i,j}$.
\end{thm}
(Robinson did not explicitly state the rate, but it follows directly from her proof, as noted in \citet{daskalakis2014counter} and explicated in our Appendix \ref{appendix:fp_worst_case_convergence_rate}.) 


\section{Anticipatory Fictitious Play}
Although FP converges to a Nash equilibrium in 2p0s games, it may take an indirect path. For example, in Rock Paper Scissors with tiebreaking towards the minimum strategy index, the sequence of average strategies $\{\xbar_1,\xbar_2,\dots\}$ orbits the Nash equilibrium, slowly spiraling in with decreasing radius, as shown on the left in Figure \ref{fig:strategy_space_rps}. This tiebreaking scheme is not special; under random tiebreaking, the path traced by FP is qualitatively the same, resulting in slow convergence with high probability, as shown in Figure \ref{fig:performance_rps}.

\begin{figure}[ht]
    \centering
    \includegraphics[width=0.75\textwidth]{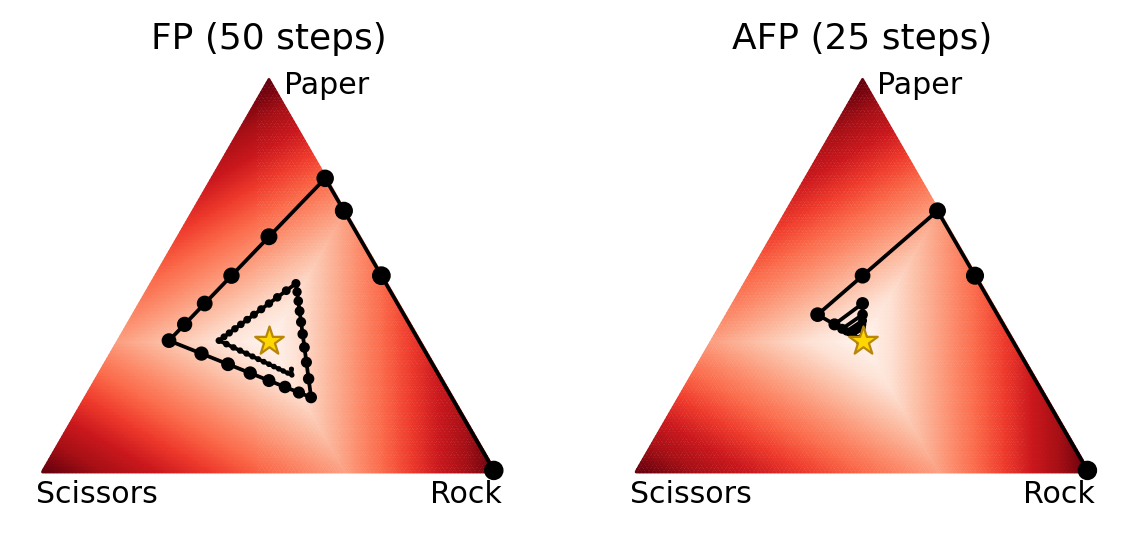}
     \caption[The first 50 steps of FP and AFP on Rock Paper Scissors]{A visualization of the first 50 steps of FP $(\xbar_1^\FP, \xbar_2^\FP, \dots, \xbar_{50}^\FP)$ and first 25 steps of AFP $(\xbar_1^\AFP, \xbar_2^\AFP, \dots, \xbar_{25}^\AFP)$ on Rock Paper Scissors. This corresponds to an equal amount of computation per algorithm (50 best responses). Ties between best response strategies were broken according to the ordering `Rock,' `Paper,' `Scissors.'  The Nash equilibrium is marked by a star. The shading indicates the exploitability of the strategy at that point, with darker colors representing greater exploitability.}
    \label{fig:strategy_space_rps}
\end{figure}

\begin{figure}[t]
    \centering
    \includegraphics[width=0.7\textwidth]{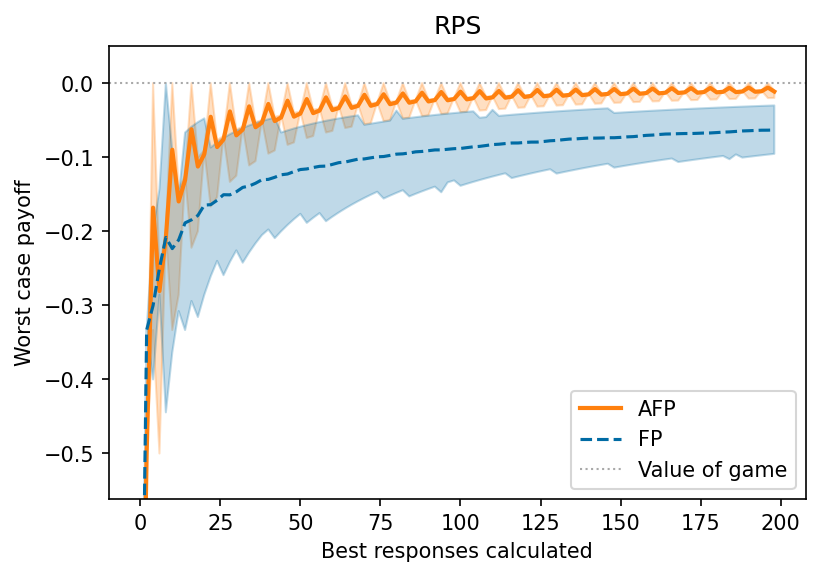}
    \caption[Comparison of FP and AFP performance on RPS with random tiebreaking.]{Comparison of FP and AFP performance ($\min \xbar_t^\T A$) on RPS with random tiebreaking. The highlighted region depicts the 10th and 90th percentiles across 10,000 runs. All variation is due to randomly sampled tiebreaking. The value of the game is $v^*=0$.}
    \label{fig:performance_rps}
    
\end{figure}

Given that FP appears to do an especially poor job of decreasing exploitability in the case above, we consider alternatives. Inspired by gradient descent, we ask if there is there a way to follow the gradient of exploitability towards the Nash equilibrium (without explicitly computing it, as is done in \citet{lockhart2019computing}). By definition, the best response to the average strategy is a strategy that maximally exploits the average strategy. So, a natural choice of update to reduce exploitability is to move the average strategy in a direction that counters this best response.

To this end, we propose {\em anticipatory fictitious play} (AFP), a version of fictitious play that ``anticipates'' the best response an adversary might play against the current strategy, and then plays the best response to an average of that and the adversary's current average strategy. (Simply responding directly to the opponent's response does not work; see Appendix \ref{sec:naive_afp}.) Alternatively, one can think of AFP as a version of FP that ``forgets'' every other best response it calculates. This latter interpretation enables a convenient implementation of AFP as a modification of FP as demonstrated in Algorithm \ref{alg:fp_afp_rl}.

\begin{remark*}
The AFP update is seemingly consistent with human psychology: it is quite intuitive to imagine how an adversary might try to exploit oneself and to respond in order to best counter that strategy. Given that fictitious play provides a model for how humans or other non-algorithmic decision makers might arrive at an equilibrium in practice \citep{luce1989games, brown1951iterative} anticipatory fictitious play offers a new model for how this may occur. We leave further consideration of this topic to future work.
\end{remark*}

AFP is given by the following process. For some $i \in \{1,\dots,m\}$ and $j \in \{1,\dots,n\}$, let $x_1 = \xbar_1 = e_i$ and $y_1 = \ybar_1 = e_j$ be initial strategies for each player. For each $t \in \N$, define
\begin{align}
    x'_{t+1} &\in \br^1_A(\ybar_{t}); & y'_{t+1} &\in \br^2_A(\xbar_{t}); \nonumber \\
    \xbar'_{t+1} &= \tfrac{t}{t+1} \xbar_{t} + \tfrac{1}{t+1} x'_{t+1}; & \ybar'_{t+1} &= \tfrac{t}{t+1} \ybar_{t} + \tfrac{1}{t+1} y'_{t+1}; \nonumber \\
    x_{t+1} &\in \br^1_A(\ybar'_{t+1}) ; & y_{t+1} &\in \br^2_A(\xbar'_{t+1}) ; \nonumber \\
    \xbar_{t+1} &= \frac{1}{t+1} \sum_{k=1}^{t+1} x_t; & \ybar_{t+1} &= \frac{1}{t+1} \sum_{k=1}^{t+1} y_t. \label{eqn:afp_update}
\end{align}
Here, $x'_{t+1}$ and $y'_{t+1}$ are the best response to the opponent's average strategy. They are the strategies that FP would have played at the current timestep. In AFP, each player ``anticipates'' this attack and defends against it by calculating the opponent's average strategy that include this attack ($\xbar'_t$ and $\ybar'_t$), and then playing the best response to the anticipated average strategy of the opponent.

In Figure \ref{fig:strategy_space_rps}, we see the effect of anticipation geometrically: AFP ``cuts corners,'' limiting the extent to which it overshoots its target. In contrast, FP aggressively overshoots, spending increasingly many steps playing strategies that take it further from its goal. The effect on algorithm performance is pronounced, with AFP hovering near equilibrium while FP slowly winds its way there.

Of course, RPS is a very specific example. It is natural to wonder: is AFP good in general? The rest of the paper seeks to answer that question. We begin by proving that AFP converges to a Nash equilibrium.

\begin{prop} \label{prop:AFP_converges}
If $\{(x_t,y_t)\}$ is an AFP process for a 2p0s game with payoff matrix $A \in \R^{m \times n}$, the conclusion of Theorem \ref{thm:FP_converges} holds for this process. Namely, AFP converges to a Nash equilibrium, and it converges no slower than the rate that bounds FP.
\end{prop}

\begin{proof}(Idea) Generalize the original proof of Theorem \ref{thm:FP_converges}. We work with accumulating payoff vectors $U(t) = t A^\T \xbar_t$ and $V(t) = t A \ybar_t$. In the original proof, a player 1 strategy index $i \in \{1,\dots,m\}$ is called {\em eligible} at time $t$ if $i \in \arg \max V(t)$ (similarly for player 2).  We replace eligibility with the notion of {\em $E$-eligibility}, satisfied by an index $i \in \arg\max [V(t) + E]$, for any $E \in \R^m$ with $\| E \|_\infty < \max_{i,j} |A_{i,j}|$. Essentially, an index is $E$-eligible if it corresponds to a best response to a perturbation of the opponent's history $\ybar_t$ or a perturbation of the game itself. The original proof structure can be preserved in light of this replacement, requiring only minor modifications to some arguments and new constants. Treating the in-between strategies in AFP, $\xbar'_t$ and $\ybar'_t$, as perturbations of $\xbar_t$ and $\ybar_t$, it follows that AFP satisfies the conditions for the generalized result. A complete proof is given in Appendix \ref{appendix:proof_of_convergence}. 
\end{proof}

\section{Application to normal form games}

Proposition \ref{prop:AFP_converges} establishes that AFP converges and that AFP's worst-case convergence rate satisfies the same bound as FP's, where the worst-case is with respect to games and tiebreaking rules. The next proposition shows that for two classes of games of interest, AFP not only outperforms FP, but attains an optimal rate. In both classes, our proofs will reveal that AFP succeeds where FP fails because AFP avoids playing repeated strategies. The results hold for general applications of FP and AFP rather than relying on specific tiebreaking rules. 

The classes of games that we analyze are intended to serve as abstract models of two fundamental aspects of real-world games: transitivity (akin to ``skillfullness;'' some ways of acting are strictly better than others) and nontransitivity (most notably, in the form of strategy cycles like Rock < Paper < Scissors < Rock). Learning algorithms for real-world games must reliably improve along the transitive dimension while accounting for the existence of strategy cycles; see \citet{balduzzi2019open} for further discussion.

For each integer $n \geq 3$, define payoff matrices $C^n$ and $T^n$ by
\begin{align*}
    C^n_{i,j} &= \begin{cases}
    \hphantom{-}1 &\text{if } i=j+1 \mod n; \\
    -1 &\text{if } i=j-1 \mod n; \\
    \hphantom{-}0 &\text{otherwise;}
    \end{cases} 
    \;\;\text{ and } \;\;
 T^n_{i,j} &= \begin{cases}
    \hphantom{-}(n-i+2)/n &\text{if } i=j+1; \\
    -(n-i+2)/n &\text{if } i=j-1; \\
    \hphantom{-(n-}\;0 &\text{otherwise,}
    \end{cases}
\end{align*}
for $i, j \in \{1,\dots,n\}$. The game given by $C^n$ is a purely cyclic game: each strategy beats the one before it and loses to the one after it; $C^3$ is the  game Rock Paper Scissors. For each $C^n$, a Nash equilibrium strategy is $[n^{-1}, \dots, n^{-1}]^\T$. The game given by $T^n$ could be considered ``transitive:'' each strategy is in some sense better than the last, and $[0, \dots, 0, 1]^\T$ is a Nash equilibrium strategy. The payoffs are chosen so that each strategy $i$ is the unique best response to $i-1$, so that an algorithm that learns by playing best responses will progress one strategy at a time rather than skipping to directly to strategy $n$ (as would happen if FP or AFP were applied to a game that is transitive in a stronger sense, such as one with a single dominant strategy; c.f. the definition of transitivity in \citep{balduzzi2019open}). Note: $T^n$ could be defined equivalently without the $n^{-1}$ factor, but this would create a spurious dependence on $n$ for the rates we derive.

The following proposition establishes a convergence rate of $O(t^{-1})$ for AFP applied to $C^n$ and $T^n$. This rate is optimal within the class of time-averaging algorithms, because the rate at which an average changes is $t^{-1}$. Note: we say a random variable $Y_t = \Omega_p(g(t))$ if, for any $\epsilon > 0$, there exists $c > 0$ such that $P[Y_t < c g(t)] < \epsilon$ for all $t$.

\begin{prop} \label{prop:comparison_on_transitive_and_cyclic}
FP and AFP applied symmetrically to $C^n$ and $T^n$ obtain the rates given in Table \ref{tab:convergence rates}. In particular, if $\{x_t, x_t\}_{t\in\N}$ is an FP or AFP process for a 2p0s game with payoff matrix $G \in \{C^n, T^n\}$ with tiebreaking as indicated, then $\max G \xbar_t = R(t).$ Tiebreaking refers to the choice of $x_{t+1} \in \arg\max \br^1_G(\xbar_t)$ when there are multiple maximizers. The ``random'' tiebreaking chooses between tied strategies independently and uniformly at random. For entries marked with ``arbitrary'' tiebreaking, the convergence rate holds no matter how tiebreaks are chosen.
\end{prop}

\begin{table}[htb]
    \centering
    \caption{Convergence rates for FP and AFP on $C^n$ and $T^n$.}
\begin{tabular}{ccccc}
    \toprule
    Algorithm & Game $G$ & Tiebreaking  & Rate $R(t)$ &Caveats  \\ \midrule
     FP  & $C^n$ & random & $\Omega_p(t^{-1/2})$  & \\ 
     AFP & $C^n$ & arbitrary & $O(t^{-1})$  &  $n=3,4$ \\
     FP & $T^n$  & arbitrary & $\Omega(t^{-1/2})$ & $t < t^*(n)$\\ 
     AFP & $T^n$ & arbitrary & $O(t^{-1})$  & \\
     \bottomrule
\end{tabular}
    
\label{tab:convergence rates}
\end{table}

\begin{proof}(Sketch, $C^n$) Full proofs of all cases are provided in Appendix \ref{appendix:convergence_rate_proofs}. Define $\Delta_0 = [0, \dots, 0]^\T \in \mathbb{Z}^n$ and $\Delta_t = t C^n \, \xbar_t$ for each $t \in \N$. The desired results are equivalent to $\max \Delta_t = O_p(\sqrt{t})$ under FP for the given tiebreaking rule, and $\max \Delta_t$ is bounded under AFP for $n=3,4$. Let $i_t$ be the index played by FP (AFP) at time $t$ (so $x_t = e_{i_t}$). It follows that
    \begin{align}
        \Delta_{t+1,j} = \begin{cases}
        \Delta_{t,j}-1 &\text{if } j=i_t-1 \mod n; \\ 
        \Delta_{t,j}+1 &\text{if } j=i_t+1 \mod n; \\
        \Delta_{t,j}  &\text{otherwise;}
        \end{cases} \label{eqn:delta_sequence_update}
    \end{align}
for each $t \in \N_0$ and $j \in \{1,\dots, n\}$. Note that the entries of $\Delta_t$ always sum to zero.

In the case of FP, it is easy to verify that $\max \Delta_t$ is nondecreasing for any choice of tiebreaking. For each $m \in \N_0$, define $t_m = \inf \{ t \in \N_0 : \max \Delta_t = m\}$. Then by Markov's inequality,
\begin{align*}
    P( \max \Delta_t < m) = P(t_m > t) \leq E(t_m)/t= \frac{1}{t} \sum_{k=0}^{m-1} E(t_{k+1} - t_{k}).
\end{align*}
Examining the timesteps at which $i_{t+1} \neq i_t$ and relating them to $\{t_k\}$, we show in the appendix that the time to increment the max from $k$ to $k+1$ satisfies $E(t_{k+1} - t_{k}) = O(k)$. Thus the bound above becomes $P(\max \Delta_t < m) \leq O(m^2)/t$. Now let $c \in \R^{\geq 0}$ be arbitrary and plug in $c \lceil \sqrt{t} \rceil$ for $m$, so we have $P(\max \Delta_t < c \sqrt{t}) \leq c^2 O(1) \rightarrow 0$ as $c \rightarrow 0$. So $\max \Delta_t = \Omega_p(\sqrt{t})$.

For the AFP case, consider the first timestep at which $\max \Delta_{t} = m+1$. Working backwards and checking cases, it can be shown that in order for the maximum value to increment from $m$ to $m+1$, there must first be a timestep where there are two non-adjacent entries of $m$ with an entry of $m-1$ between them. This cannot happen in the $n=3, m=2$ case because three positive entries (2,1,2) don't sum to zero. Similarly, in the $n=4, m=2$ case, it turns out by  \eqref{eqn:delta_sequence_update} that $\Delta_t = [a, b, -a, -b]$ for some $a$, $b$. So there cannot be three positive entries in this case either. Therefore $\max_t \Delta_t \leq 2$ for $n=3,4$.
\end{proof}

The proofs of Proposition \ref{prop:comparison_on_transitive_and_cyclic} establish a theme: FP can be slow because it spends increasingly large amounts of time progressing between strategies (playing $x_t = x_{t+1} = \dots = x_{t+k}$ with $k$ increasing as $t$ increases), whereas AFP avoids this. (Return to Figure \ref{fig:strategy_space_rps} for a visual example.)

Some further comments on the results: we only obtain the $O(t^{-1})$ rate for AFP applied to $C^n$ in the $n=3,4$ case. We conjecture that: (i) for a specific tiebreaking rule, AFP has the same worst-case rate as FP but with a better constant, (ii) under random tiebreaking, AFP is $O_p(t^{-1})$ for all $n$. This is reflected in numerical simulations for large $n$, as shown in Appendix \ref{appendix:additional_figures}.

Our results are noteworthy for their lack of dependence on tiebreaking: worst-case analyses of FP typically rely on specific tiebreaking rules; see \citet{daskalakis2014counter}, for example. As for the ``$t < t^*(n)$'' caveat for FP applied to $T^n$, this is an unremarkable consequence of analyzing a game with a pure strategy equilibrium (all probability assigned to a single strategy). We write $t^*(n)$ to indicate the first index at which FP plays $e_n$. Both FP and AFP will play $e_n$ forever some finite number of steps after they play it for the first time, thus attaining a $t^{-1}$ rate as the average strategy ``catches up'' to $e_n$. Our result shows that until this point, FP is slow, whereas AFP is always fast. As before, AFP's superior performance is reflected in numerical simulations, as shown in Appendix \ref{appendix:additional_figures}.

\subsection{Numerical results} \label{sec:numerical_results}

In order to compare FP and AFP more generally, we sample large numbers of random payoff matrices and compute aggregate statistics across them. Matrix entries are sampled as independent, identically distributed, standard Gaussian variables (note that the shift- and scale-invariance of matrix game equilibria implies that the choice of mean and variance is inconsequential). Since FP and AFP are so similar, and AFP computes two best responses per timestep, it's natural to wonder: is AFP's superior performance just an artifact of using more computation per timestep? So, in order to make a fair comparison, we compare the algorithms by {\em the number of best responses calculated} instead of the number of timesteps (algorithm iterations). Using the worst-case payoff as the measure of performance, we compare FP and AFP based on the number of responses computed and based on matrix size in Figures  \ref{fig:proportion_afp_better_fp} and \ref{fig:performance_by_size}.

The result is that AFP is clearly better on both counts. Although FP is better for a substantial proportion of $30 \times 30$ games at very early timesteps $t$, AFP quickly outpaces FP, eventually across each of 1,000 matrices sampled. In terms of matrix size, FP and AFP appear equivalent on average for small matrices, but quickly grow separated as matrix size grows, with AFP likely to be much better.

\begin{figure}[ht]
    \centering
    \includegraphics[width=0.6\textwidth]{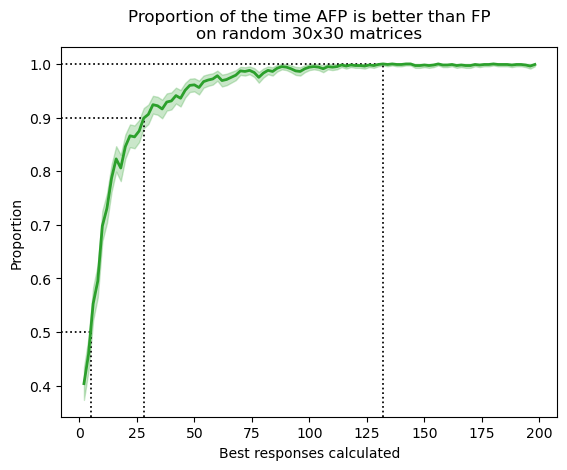}
    \caption[The proportion of the time that AFP outperforms FP on random (30,30) matrices]{For 1,000 randomly sampled (30,30) matrices $A$, the proportion of the time that $\min \, (\xbar_{r/2}^\textnormal{AFP})^\T A \geq \min \, (\xbar_{r}^\textnormal{FP})^\T A$ for $r=2,4\dots,200$. A 95\% Agresti-Coull confidence interval \citep{agresti1998approximate} for the true proportion is highlighted. 
    Note that after only about six best responses, AFP is better half the time, and by 130, AFP is better than FP essentially 100\% of the time.}
    \label{fig:proportion_afp_better_fp}
\end{figure}

\begin{figure}[ht]
    \centering
    \includegraphics[width=0.6\textwidth]{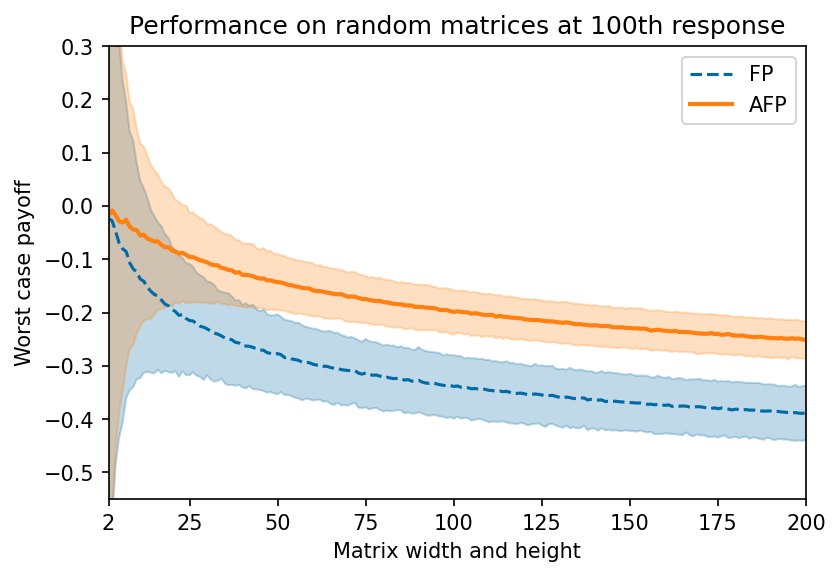}
    \caption[Average performance of FP vs. AFP at the 100th best response for randomly generated matrices]{Average performance of FP vs. AFP at the 100th best response (timestep 100 for FP, timestep 50 for AFP) as matrix size is varied. All matrices are square. Highlighted regions show the 10th and 90th percentiles.}
    \label{fig:performance_by_size}
\end{figure}


\section{Application to reinforcement learning}

We apply reinforcement learning (RL) \citep{sutton2018reinforcement} versions of FP and AFP in the context of a (two-player, zero-sum, symmetric) {\em stochastic game} (\citet{shapley1953stochastic}), defined by the tuple $(\St, \Obs, \X, \A, \Prob, \mathcal{R}, p_0)$, where $\St$ is the set of possible states of the environment, $\Obs$ is the set of possible observations received by an agent, $\X: \St \rightarrow \Obs \times \Obs$ gives the observations for each player based on the current state, $\A$ is the set of available actions, $\Prob : \St \times \A \times \A \rightarrow \Delta(\St)$ defines the transition dynamics for the environment given each player's action, $\mathcal{R} : \St \rightarrow \R \times \R$ defines the reward for both players such that $\mathcal{R}(s_t) = (r_t, -r_t)$ are the rewards observed by each player at time $t$, and $p_0 \in \Delta(\St)$ is the initial distribution of states, such that $s_0 \sim p_0$. Let $\Hist$ be the set of possible sequences of observations. Then a {\em policy} is a map $\pi : \Hist \rightarrow \Delta(A)$. An {\em episode} is played by iteratively transitioning by the environment according to the actions sampled from each players' policies at each state. Players 1 and 2 earn {\em returns} $(\sum_t r_t, -\sum_t r_t)$. The reinforcement learning algorithms we consider take sequences of observations, actions, and rewards from both players and use them to incrementally update policies toward earning greater expected returns. For background on reinforcement learning, see \citet{sutton2018reinforcement}. For details on machine learning approximations to FP, see \citet{heinrich2015fictitious}. Table \ref{tab:game_theory_to_rl_analogy} gives a high-level overview of the relationship. 

\begin{table}[htb]
    \centering
    \small
    \caption{The normal-form game analogies used to extend FP and AFP to reinforcement learning.}
    \begin{tabular}{p{4.3cm}p{8cm}}
    \toprule
    Normal-form game & Stochastic/extensive-form game \\
    \midrule
    Strategy & Policy \\
    Payoff $A_{i,j}$ & Expected return $E_{\pi_i, \pi_j}(\sum_t r_t)$ \\
    Best response & Approximate best response by RL \\
    Strategy mixture $\sum \alpha_i x_i$, $\sum \alpha_i =1$, $\alpha_i \geq 0$  & At start of episode, sample policy $\pi_i$ with probability $\alpha_i$. Play entire episode with $\pi_i$. \\
    \bottomrule
    \end{tabular}
    \label{tab:game_theory_to_rl_analogy}
\end{table}

We use two environments, our own TinyFighter, and Running With Scissors, from \citet{vezhnevets2020options}.

\subsection{Environments}
{\bf TinyFighter} is a minimal version of an arcade-style fighting game shown in Figure \ref{fig:tinyfighter_env}. It features two players with four possible actions: \texttt{Move} \texttt{Left}, \texttt{Move} \texttt{Right}, \texttt{Kick}, and \texttt{Do} \texttt{Nothing}. Players are represented by a rectangular body and when kicking, extend a rectangular leg towards the opponent.

Kicking consists of three phases: Startup, Active, and Recovery. Each phase of a kick lasts for a certain number of frames, and if the Active phase of the kick intersects with any part of the opponent (body or leg), a hit is registered. When a hit occurs, the players are pushed back, the opponent takes damage, and the opponent is stunned (unable to take actions) for a period of time. In the Startup and Recovery phases, the leg is extended, and like the body, can be hit by the opponent if the opponent has a kick in the active phase that intersects the player. The game is over when a player's health is reduced to zero or when time runs out. 

Player observations are vectors in $\R^{13}$ and contain information about player and opponent state: position, health, an `attacking' indicator, a `stunned' indicator, and how many frames a player has been in the current action. The observation also includes the distance between players, time remaining, and the direction of the opponent (left or right of self). The game is partially observable, so information about the opponent's state is hidden from the player for some number of frames (we use four, and the game runs at 15 frames per second). This means a strong player must guess about the distribution of actions the opponent may have taken recently and to respond to that distribution; playing deterministically will allow the opponent to exploit the player and so a stochastic strategy is required to play well. 

{\bf Running With Scissors} (RWS) is a spatiotemporal environment with partial observability with potential for nontransitive relationships between policies. As shown in Figure \ref{fig:running_with_scissors_env}, RWS is a 2D gridworld with a few types of entities: two agents; three types of items: rock, paper, and scissors, which can be picked up by the agents to add to their inventories; and impassable walls. In addition to moving around and picking up items, agents in RWS have a ``tag'' action, which projects a cone in front of them for a single frame. If the cone hits the other agent, the episode ends and each agent receives rewards based on the payoff matrix $C^n$ according to the ratios of each item in their inventory. Agents can only see their own inventory and a $5 \times 5$ grid in front of them and can remember the last four frames they've seen, so in order to perform effectively they must infer what items the opponent has picked up. \citet{vezhnevets2020options} and \citet{liu2022neupl} (Appendix B.1.) feature further discussion of the environment.

\begin{figure}[htb]
    \centering
    \begin{subfigure}[]{0.45\textwidth}
    \centering
     \includegraphics[width=0.87 \textwidth]{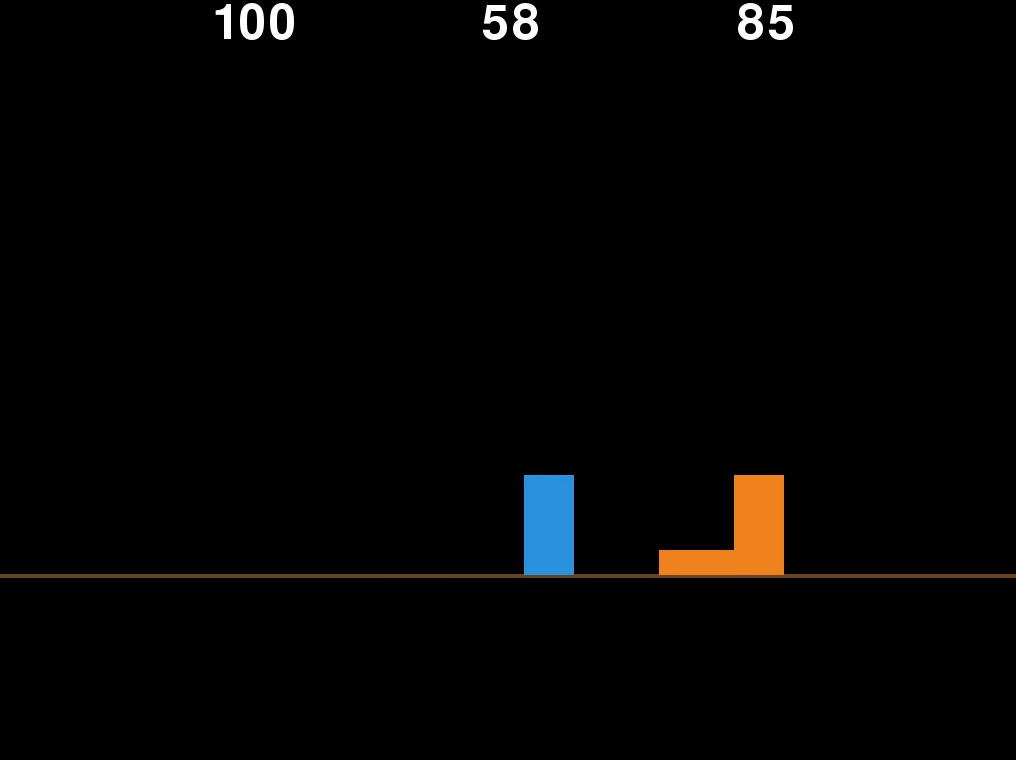}
    \caption{TinyFighter, where players dance back and forth, attempting to land kicks on each other in order to reduce the other's health from 100 to zero.}
    \label{fig:tinyfighter_env}
    \end{subfigure}
    \begin{subfigure}[]{0.45\textwidth}
    \centering
    \includegraphics[width=\textwidth]{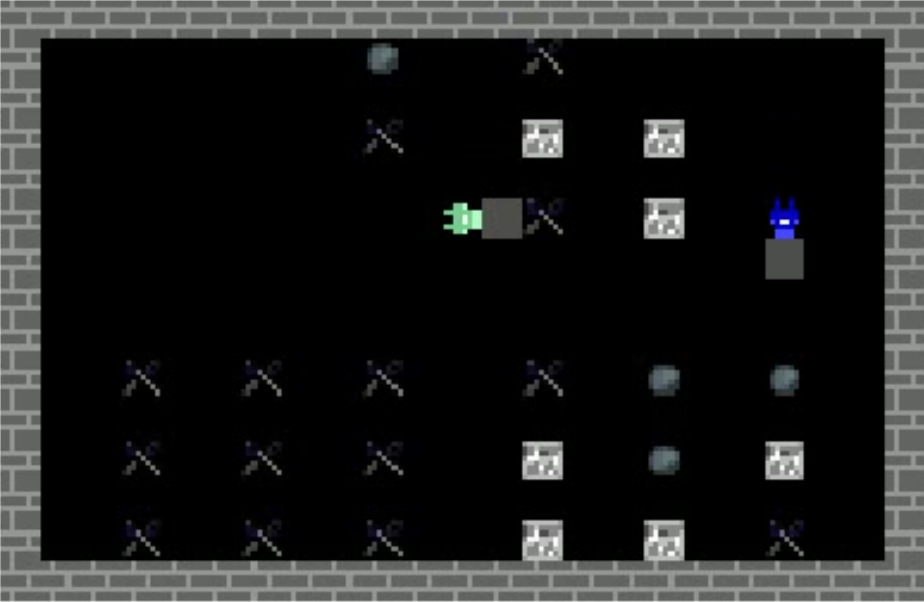}
    \caption{Running With Scissors, a 2D gridworld where players collect items before tagging each other to earn reward based on the proportion of items in their inventory.}
     \label{fig:running_with_scissors_env}
    \end{subfigure}
    
    \caption{Screenshots of multiagent RL environments.}
\end{figure}



\subsection{Adapting FP and AFP to reinforcement learning}
Neural Population Learning (NeuPL) \citep{liu2022neupl} is a framework for multiagent reinforcement learning wherein a collection of policies is learned and represented by a single neural network and all policies train continuously. 
For our experiments, we implement FP and AFP within NeuPL, as shown in Algorithm \ref{alg:fp_afp_neupl}. For reference, we also include a simple RL version of FP and AFP in the style of PSRO in Appendix \ref{appendix:vanilla_rl_fp_afp}. 

\begin{algorithm}[htb]
\caption{NeuPL-FP/AFP}
\label{alg:fp_afp_neupl}
\begin{algorithmic}[1]
\State $\mathfrak{O} \in \{ \mathfrak{O}^\text{FP}, \mathfrak{O}^\text{AFP}\}$ \Comment{Input: FP or AFP opponent sampler.}
\State $\left\{\Pi_\theta(t) : \Hist \rightarrow \Delta(\A)\right\}_{t=1}^n$ \Comment{Input:  neural population net.}
\For{Batch $b = 1, 2, 3, \dots,$}
    \State $B \gets \{ \}$
    \While {per-batch compute budget remains}
        \State $T_\text{learner} \sim \text{Uniform}(\{1,\dots,n\})$
        \State $T_\text{opponent} \sim \mathfrak{O}(T_\text{learner})$
        \State $D_\text{learner} \gets \textsc{PlayEpisode}\!\left(\Pi_\theta(T_\text{learner}), \Pi_\theta(T_\text{opponent})\right)$
        \State $B \gets B \cup D_\text{learner}$
    \EndWhile
    \State $\Pi_\theta \gets \textsc{ReinforcementLearningUpdate}(B)$
\EndFor
\end{algorithmic}
\end{algorithm}

\begin{figure}[ht]
    \centering
    \includegraphics[width=0.9\textwidth]{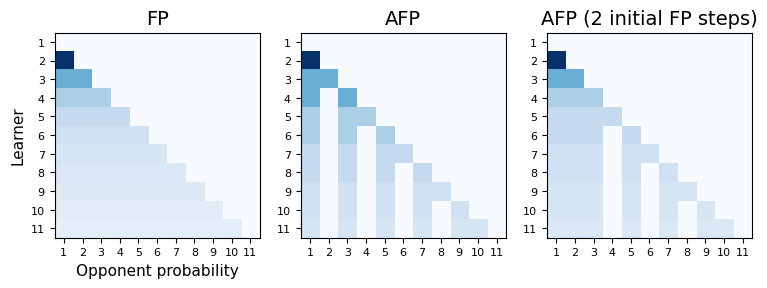}
    \caption{A visual depiction of the distributions of opponents (``meta-strategies'' in PSRO or NeuPL) each learner faces in a population learning implementation of FP or AFP. The $(i,j)$ entry is the probability that, given that agent $i$ is training, it will face agent $j$ in a particular episode. Dark blue indicates probability 1, white indicates probability 0.}
    \label{fig:metagraphs}
\end{figure}

In NeuPL-FP/AFP, the opponent sampler $\mathfrak{O}$ 
determines the distributions of opponents that each agent faces and is the only difference between the FP and AFP implementations. We have, for each $t >1$,
\begin{align*}
    \mathfrak{O}^\text{FP}(t) &= \text{Uniform}(\{1,2,3,\dots,t-1\}), \text{ and} \\
    \mathfrak{O}^\text{AFP}(t) &= \text{Uniform}(\{ k < t : k \text{ odd} \} \cup \{t-1\} ).
\end{align*}
These distributions are depicted in Figure \ref{fig:metagraphs}. Just as each step of FP involves computing a best response to an average against all prior strategies, sampling from $\mathfrak{O}^\text{FP}(t)$ corresponds to training agent $t$ uniformly against the prior policies; just as AFP can be thought of ``forgetting'' every other index, $\mathfrak{O}^\text{AFP}(t)$ trains learner index $t$ uniformly against every odd indexed policy plus the most recent policy.  The neural population net $\Pi_\theta(t) : \Hist \rightarrow \Delta(A)$ defines a different policy for each agent index $t$, and can equivalently be represented as $\Pi_\theta(a|s,t)$.

\subsection{Experimental setup}
For the neural population net, we use an actor-critic \citep{sutton2018reinforcement} architecture similar to that used for RWS in \citet{liu2022neupl}: first, a module is used to process environment observations into a dense representation that is shared between an actor head and a critic head. The actor takes this vector, concatenates it with a vector representing the distribution of opponents faced by the currently training agent $t$ (e.g., $[0.5, 0.5, 0, \dots, 0]$ for agent $t=3$ in FP or AFP), then processes it with a dense MLP with ReLu activations, with action masking applied prior to a final Softmax layer. The critic is similar, except an additional input is used: the index of the opponent sampled at the beginning of this episode, matching the original implementation.\footnote{Note that the actor (policy network) does not observe which opponent it faces, only the {\em distribution} over agents it faces; this is important because otherwise our agent would not learn a best response to an average policy as intended in FP and AFP. The reason for including this information for the critic (value network) is that it may reduce the variance of the value function estimator.} For the exact architectures used, see Appendix \ref{appendix:rl_experiment_hyperparameters}. We use a neural population size of $n=11$. Based on matrix game simulations and a preliminary RL experiments, we determined that AFP performs slightly better in a short time horizon when initialized with two steps of FP, so we do this, as shown in the final panel of Figure \ref{fig:metagraphs}. See Appendix \ref{appendix:additional_figures} for a figure comparing performance in the matrix setting.

We implemented NeuPL within a basic self-play reinforcement learning loop by wrapping the base environment (TinyFighter or RWS) within a lightweight environment that handles NeuPL logic, such as opponent sampling. For reinforcement learning, we use the Asynchronous Proximal Policy Optimization (APPO) algorithm \citep{schulman2017proximal}, a distributed actor-critic RL algorithm, as implemented in RLLib \citep{moritz2018ray} with a single GPU learner. Hyperparameter settings are given in Appendix \ref{appendix:rl_experiment_hyperparameters}. We train the entire neural population net (agents 1-11) for 12,500 steps, where a step is roughly 450 minibatch updates of stochastic gradient descent. This corresponds to about five days of training. For each of the two environments, we repeat this procedure independently 10 times for FP and 10 times for AFP.

\subsection{Results}
To evaluate exploitability, we made use of the fact that each FP and AFP neural population are made up of agents trained to ``exploit'' the ones that came before them. Specifically, each agent is trained to approximate a best response to the average policy returned by the algorithm at the previous timestep. So, to estimate the exploitability of NeuPL-FP or NeuPL-AFP at step $t \in \{1,\dots,n-1\}$, we simply use the average return earned by agent $t+1$ against agents $\{1,\dots,t\}$ to obtain the {\em within-population exploitability} at $t$. This is a convenient metric, but insufficient on its own. In order for it to be meaningful, the agents in the population must have learned approximate best responses that are close to actual best responses; if they have not, it could be that within-population exploitability is low, not because the average policy approximates a Nash policy but because nothing had been learned at all. To account for this, we also evaluate the populations learned using {\em relative population performance} \citep{balduzzi2019open}, which measures the strength of one population of agents against the other. The purpose of using relative population performance is simply to verify that one algorithm did not produce generally more competent agents than the other.

We paired each of the 10 replicates of FP and AFP and computed the relative population performance for each. On TinyFighter, the average was -0.73, with a Z-test-based 90\% confidence interval width of 1.37. On RWS, the average was 4.60, with a Z-test-based 90\% confidence interval width of 3.50. We conclude that the agents learned by FP and AFP are not statistically significantly different in terms of performance for TinyFighter, but the agents learned by FP have a slight, statistically significant advantage in RWS. However, these differences are small relative to the total obtainable reward in either environment (20 for TinyFighter, roughly 60 for RWS), so we conclude it is reasonable to use within-population exploitability to compare FP and AFP, as shown in Figure \ref{fig:rl_results}. For consistency with the matrix game simulation results, we plot worst case payoff, which is simply the negative of exploitability in this case.

\begin{figure}
\centering
\begin{subfigure}[]{0.45\textwidth}
    \centering
    \includegraphics[width=\textwidth]{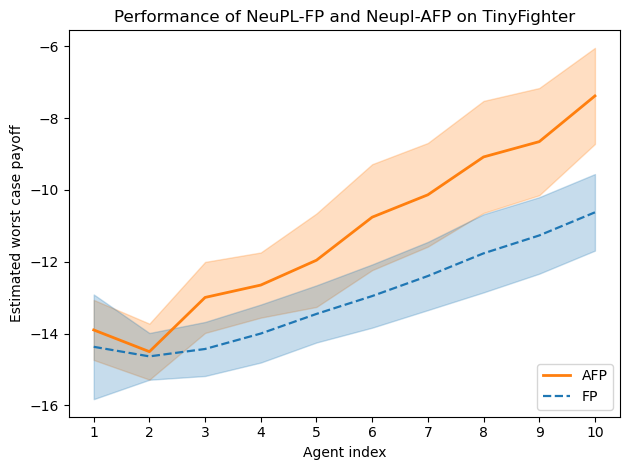}
    \caption[Comparison of NeuPL-FP and NeuPL-AFP on TinyFighter]{TinyFighter}
    \label{fig:tinyfighter_exploitability}
 \end{subfigure}
\begin{subfigure}[]{0.45\textwidth}
    \centering
    \includegraphics[width=\textwidth]{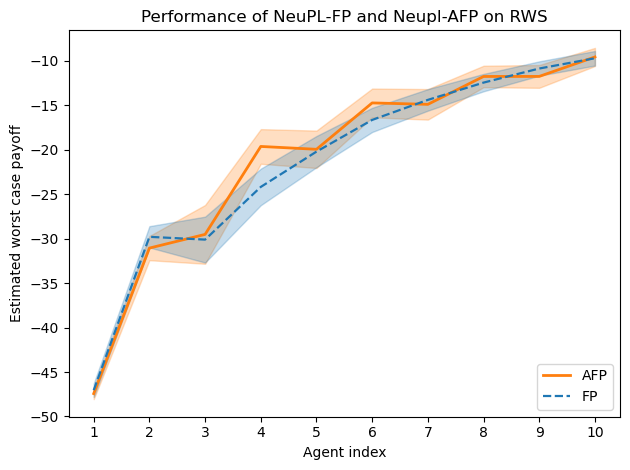}
    \caption[Comparison of NeuPL-FP and NeuPL-AFP on TinyFighter]{Running With Scissors}
    \label{fig:rws_exploitability}
\end{subfigure}
\caption{Estimated worst-case payoffs for FP and AFP on two stochastic games. Highlighting indicates a pointwise 90\% confidence region.}
\label{fig:rl_results}
\end{figure}

We find that AFP has a significantly better worst-case payoff of -7.4 versus 10.6 for FP at the final timestep in TinyFighter. This corresponds to a noteworthy 16\% reduction in exploitability relative to the total possible reward of 20 that can be earned in TinyFighter. In RWS, the algorithms have essentially identical performance. The fact that AFP's advantage varies widely by environment is not surprising. The matrix game simulations in Figure \ref{fig:proportion_afp_better_fp} showed that until over 100 steps of each algorithm, there is some proportion of games for which FP performs better. Correspondingly, we would expect that there is a nontrivial proportion of stochastic games where NeuPL-FP outperforms NeuPL-AFP for small population sizes. Although we expect NeuPL will not be able to support over 100 policies (the original paper used population size 8), it would be possible to do so within the PSRO framework. This remains a topic for further investigation.


\section{Conclusion}
We proposed a variant of fictitious play for faster estimation of Nash equilibria in two-player, zero-sum games. Anticipatory fictitious play is intuitive, easy to implement, and supported by theory and numerical simulations which suggest that it is virtually always preferable to fictitious play. Consequently, we shed new light on two motivating problems for fictitious play: primarily, large-scale multiagent reinforcement learning for complicated real-world games; also, modeling strategic decision making in humans. Further work is needed to understand the conditions under which AFP outperforms FP in the reinforcement learning setting.




\section*{Acknowledgements}
The authors sincerely thank Ryan Martin for continued mentorship, especially around theory and the presentation of the paper; Philip Wardlaw for orchestrating preliminary reinforcement learning experiments; Adam Venis and Angelo Olcese for contributions to the proofs of FP and AFP applied to $C^n$; Jesse Clifton and Eric Laber for helpful comments on a draft;  Jesse Clifton and Marc Lanctot for suggesting related works that had been overlooked; and Siqi Liu for helpful explanations and implementation details for NeuPL.
 


\bibliography{refs}

\begin{thebibliography}{}

\bibitem[Adler, 2013]{adler2013equivalence}
Adler, I. (2013).
\newblock The equivalence of linear programs and zero-sum games.
\newblock {\em International Journal of Game Theory}, 42(1):165--177.

\bibitem[Agresti and Coull, 1998]{agresti1998approximate}
Agresti, A. and Coull, B.~A. (1998).
\newblock Approximate is better than “exact” for interval estimation of
  binomial proportions.
\newblock {\em The American Statistician}, 52(2):119--126.

\bibitem[Balduzzi et~al., 2019]{balduzzi2019open}
Balduzzi, D., Garnelo, M., Bachrach, Y., Czarnecki, W., Perolat, J., Jaderberg,
  M., and Graepel, T. (2019).
\newblock Open-ended learning in symmetric zero-sum games.
\newblock In {\em International Conference on Machine Learning}, pages
  434--443. PMLR.

\bibitem[Brown, 1951]{brown1951iterative}
Brown, G.~W. (1951).
\newblock Iterative solution of games by fictitious play.
\newblock {\em Activity Analysis of Production and Allocation}, 13(1):374--376.

\bibitem[Chollet et~al., 2015]{chollet2015keras}
Chollet, F. et~al. (2015).
\newblock Keras.

\bibitem[Conlisk, 1993a]{conlisk1993adaptation}
Conlisk, J. (1993a).
\newblock Adaptation in games: Two solutions to the crawford puzzle.
\newblock {\em Journal of Economic Behavior \& Organization}, 22(1):25--50.

\bibitem[Conlisk, 1993b]{conlisk1993adaptive}
Conlisk, J. (1993b).
\newblock Adaptive tactics in games: Further solutions to the crawford puzzle.
\newblock {\em Journal of Economic Behavior \& Organization}, 22(1):51--68.

\bibitem[Daskalakis and Pan, 2014]{daskalakis2014counter}
Daskalakis, C. and Pan, Q. (2014).
\newblock A counter-example to karlin's strong conjecture for fictitious play.
\newblock In {\em 2014 IEEE 55th Annual Symposium on Foundations of Computer
  Science}, pages 11--20, Philadelphia, PA, USA. IEEE, IEEE.

\bibitem[Foerster et~al., 2018]{foerster2018learning}
Foerster, J., Chen, R.~Y., Al-Shedivat, M., Whiteson, S., Abbeel, P., and
  Mordatch, I. (2018).
\newblock Learning with opponent-learning awareness.
\newblock In {\em Proceedings of the 17th International Conference on
  Autonomous Agents and MultiAgent Systems}, pages 122--130, '. AAMAS.

\bibitem[Heinrich et~al., 2015]{heinrich2015fictitious}
Heinrich, J., Lanctot, M., and Silver, D. (2015).
\newblock Fictitious self-play in extensive-form games.
\newblock In {\em International Conference on Machine Learning}, pages
  805--813.

\bibitem[Kuhn, 1953]{kuhn1953extensive}
Kuhn, H. (1953).
\newblock Extensive games and the problem of information.
\newblock In {\em Contributions to the Theory of Games}, pages 193--216.
  Princeton University Press.

\bibitem[Lanctot et~al., 2017]{lanctot2017unified}
Lanctot, M., Zambaldi, V., Gruslys, A., Lazaridou, A., Tuyls, K., Perolat, J.,
  Silver, D., and Graepel, T. (2017).
\newblock A unified game-theoretic approach to multiagent reinforcement
  learning.
\newblock In Guyon, I., Luxburg, U.~V., Bengio, S., Wallach, H., Fergus, R.,
  Vishwanathan, S., and Garnett, R., editors, {\em Advances in Neural
  Information Processing Systems}, volume~30, -. Curran Associates, Inc.

\bibitem[Letcher et~al., 2018]{letcher2018stable}
Letcher, A., Foerster, J., Balduzzi, D., Rockt{\"a}schel, T., and Whiteson, S.
  (2018).
\newblock Stable opponent shaping in differentiable games.
\newblock In {\em International Conference on Learning Representations},
  page~', Vancouver Convention Center, Vancouver, BC, Canada. ICLR.

\bibitem[Liu et~al., 2022]{liu2022neupl}
Liu, S., Marris, L., Hennes, D., Merel, J., Heess, N., and Graepel, T. (2022).
\newblock Neupl: Neural population learning.
\newblock In {\em The Tenth International Conference on Learning
  Representations, {ICLR} 2022, Virtual Event, April 25-29, 2022}.
  OpenReview.net.

\bibitem[Lockhart et~al., 2019]{lockhart2019computing}
Lockhart, E., Lanctot, M., P{\'e}rolat, J., Lespiau, J.-B., Morrill, D.,
  Timbers, F., and Tuyls, K. (2019).
\newblock Computing approximate equilibria in sequential adversarial games by
  exploitability descent.
\newblock {\em arXiv preprint arXiv:1903.05614}, '(').

\bibitem[Luce and Raiffa, 1989]{luce1989games}
Luce, R.~D. and Raiffa, H. (1989).
\newblock {\em Games and Decisions: Introduction and Critical Survey}.
\newblock Courier Corporation, '.

\bibitem[Moritz et~al., 2018]{moritz2018ray}
Moritz, P., Nishihara, R., Wang, S., Tumanov, A., Liaw, R., Liang, E., Elibol,
  M., Yang, Z., Paul, W., Jordan, M.~I., and Stoica, I. (2018).
\newblock Ray: A distributed framework for emerging $\{$AI$\}$ applications.
\newblock In {\em 13th USENIX Symposium on Operating Systems Design and
  Implementation (OSDI 18)}, pages 561--577.

\bibitem[Nash~Jr, 1950]{nash1950equilibrium}
Nash~Jr, J.~F. (1950).
\newblock Equilibrium points in n-person games.
\newblock {\em Proceedings of the National Academy of Sciences}, 36(1):48--49.

\bibitem[Robinson, 1951]{robinson1951iterative}
Robinson, J. (1951).
\newblock An iterative method of solving a game.
\newblock {\em Annals of Mathematics}, 54(2):296--301.

\bibitem[Schulman et~al., 2017]{schulman2017proximal}
Schulman, J., Wolski, F., Dhariwal, P., Radford, A., and Klimov, O. (2017).
\newblock Proximal policy optimization algorithms.

\bibitem[Shamma and Arslan, 2005]{shamma2005dynamic}
Shamma, J.~S. and Arslan, G. (2005).
\newblock Dynamic fictitious play, dynamic gradient play, and distributed
  convergence to nash equilibria.
\newblock {\em IEEE Transactions on Automatic Control}, 50(3):312--327.

\bibitem[Shapley, 1953]{shapley1953stochastic}
Shapley, L.~S. (1953).
\newblock Stochastic games.
\newblock {\em Proceedings of the National Academy of Sciences},
  39(10):1095--1100.

\bibitem[Shoham and Leyton-Brown, 2008]{shoham2008multiagent}
Shoham, Y. and Leyton-Brown, K. (2008).
\newblock {\em Multiagent Systems: Algorithmic, Game-theoretic, and Logical
  Foundations}.
\newblock Cambridge University Press.

\bibitem[Sutton and Barto, 2018]{sutton2018reinforcement}
Sutton, R.~S. and Barto, A.~G. (2018).
\newblock {\em Reinforcement Learning: An Introduction}.
\newblock MIT press, Cambridge, Massachusetts.

\bibitem[Vezhnevets et~al., 2020]{vezhnevets2020options}
Vezhnevets, A., Wu, Y., Eckstein, M., Leblond, R., and Leibo, J.~Z. (2020).
\newblock Options as responses: Grounding behavioural hierarchies in
  multi-agent reinforcement learning.
\newblock In {\em International Conference on Machine Learning}, pages
  9733--9742, '. PMLR, PMLR.

\bibitem[Vinyals et~al., 2019]{vinyals2019grandmaster}
Vinyals, O., Babuschkin, I., Czarnecki, W.~M., Mathieu, M., Dudzik, A., Chung,
  J., Choi, D.~H., Powell, R., Ewalds, T., Georgiev, P., Oh, J., Horgan, D.,
  Kroiss, M., Danihelka, I., Huang, A., Sifre, L., Cai, T., Agapiou, J.~P.,
  Jaderberg, M., Vezhnevets, A.~S., Leblond, R., Pohlen, T., Dalibard, V.,
  Budden, D., Sulsky, Y., Molloy, J., Paine, T.~L., Gulcehre, C., Wang, Z.,
  Pfaff, T., Wu, Y., Ring, R., Yogatama, D., McKinney, K., Smith, O., Schaul,
  T., Lillicrap, T., Kavukcuoglu, K., Hassabis, D., Apps, C., and Silver, D.
  (2019).
\newblock Grandmaster level in starcraft ii using multi-agent reinforcement
  learning.
\newblock {\em Nature}, 575(7782):350--354.

\bibitem[Zhang and Lesser, 2010]{zhang2010multi}
Zhang, C. and Lesser, V. (2010).
\newblock Multi-agent learning with policy prediction.
\newblock In {\em Twenty-fourth AAAI conference on artificial intelligence},
  pages 927--937, '. AAAI.

\end{thebibliography}


\clearpage

\appendix

\noindent{\huge\bfseries \centering Appendix\par}
\section{Why naive AFP doesn't work} \label{sec:naive_afp}
The original idea for the AFP algorithm, which we refer to as {\em naive AFP}, was: at each timestep, play the best response to the opponent's best response to your history. Formally, this given by the updates (c.f. the AFP update \eqref{eqn:afp_update}):
\begin{align*}
    x'_{t+1} &\in \br^1_A(\ybar_t); & y'_{t+1} &\in \br^2_A(\xbar_t); \nonumber \\
    x_{t+1} &\in \br^1_A(y'_{t+1}) ; & y_{t+1} &\in \br^2_A(x'_{t+1}) ; \nonumber \\
    \xbar_{t+1} &= \frac{1}{t+1} \sum_{k=1}^{t+1} x_k; & \ybar_{t+1} &= \frac{1}{t+1} \sum_{k=1}^{t+1} y_k.
\end{align*}
In preliminary simulations, naive AFP performed well in cyclic games and seemed to be competitive with FP in other games. 

However, upon inspection it becomes clear that Naive AFP is not guaranteed to converge to a Nash equilibrium. This is because Naive AFP can only play best responses to strategies returned by the best response operator, which are pure strategies, but Nash equilibria of some games have nonzero probability assigned to strategies that are not best responses to any pure strategies. Thus there are some games where naive AFP is incapable of assigning any probability to actions which must be assigned nonzero probability in a Nash equilibrium, so naive AFP cannot converge to a Nash equilibrium. For example, see the game given in Table \ref{tab:RPS_safeRock}.

\begin{table}[ht]
    \centering
    \begin{tabular}{c|ccc}
                  & Rock & Paper & Scissors  \\ \hline
        Rock      & 0    & -1    & 1         \\            
        Paper     & 1    & 0     & -1        \\   
        Scissors  & -1   & 1     & 0         \\    
        SafeRock & 0    & 0     & 0.99      \\  
    \end{tabular}
    \caption{Rock Paper Scissors SafeRock. In this game, SafeRock is not a best response to any of the opponent's pure strategies, so it won't be played by naive AFP. However, SafeRock is included in the Nash equilibrium which is given approximately by the following probabilities $([0.33,0.47,0.2],[0,0.32,0.18,0.15])$, with value $v^* = 0.133$. If SafeRock were removed the game would be symmetric and have value 0, so we know SafeRock must be included in {\em any} Nash support.}
    \label{tab:RPS_safeRock}
\end{table}


\section[Proof that AFP converges]{Proof of Proposition~\ref{prop:AFP_converges}} \label{appendix:proof_of_convergence}
Our analysis closely follows the original proof of FP's convergence given in \citet{robinson1951iterative}, which consists of four key lemmas. We introduce the notion of a perturbed fictitious play system, a slight generalization of Robinson's ``vector system.'' We modify Robinson's second lemma accordingly, which causes inconsequential changes in the third and fourth lemma, leaving the final result the same. In this way, we extend Robinson's proof to a broader class of algorithms which include fictitious play and anticipatory fictitious play as special cases.

\begin{defn} An {\em iterative play system} $(U,V)$ for $A \in \R^{m \times n}$ is a pair of sequences of vectors $U=\{U(0),U(1),
\dots\}$ and $V=\{V(0),V(1),\dots\}$ with each $U(t) \in \R^n$ and $V(t) \in \R^m$ satisfying the following properties:
\begin{itemize}
    \item $\min U(0) = \max V(0)$, and 
    \item for each $t$, $U(t+1) = U(t) + A_{i(t),*}$ and $V(t+1) = V(t) + A_{*,j(t)}$,
\end{itemize}
where $A_{i,*}$ and $A_{*,j}$ are the $i$th row and $j$th column of $A$, and $i(t)$ and $j(t)$ are the row and column ``played'' by players 1 and 2 at time $t$.
\end{defn}

The interpretation of an iterative play system is as follows. Suppose we choose $U(0) = 0$ and $V(0) = 0$. Write $e_h$ to indicate a vector with a 1 at index $h$ and 0's elsewhere. Then $\xbar_t = t^{-1} \sum_{k=1}^t e_{i(k)}$ and $\ybar_t = t^{-1} \sum_{k=1}^t e_{j(k)}$ are the empirical strategies played by players 1 and 2, and $t^{-1} V(t) = A \ybar_t$ and $t^{-1} U(t) = A^\T \xbar_t$ are the payoffs for player 1 and 2 when faced with those empirical strategies. In this way, $V(t)$ and $U(t)$ can be seen as ``accumulating empirical payoffs'' for players $1$ and $2$.

\begin{defn} A {\em perturbed fictitious play} system (PFP-system) is an iterative play system with the additional property that $i(t)$ and $j(t)$ are best responses to a perturbed set of empirical payoffs. Precisely, an iterative play system is a PFP-system such that for any values $E_V(t) \in \R^m$ and $E_U(t) \in \R^n$ with $\|E_V(t)\|_\infty < a := \max_{i,j} |a_{i,j}|$ and $\|E_U(t)\|_\infty < a$ for each $t$, 
\begin{align*}
    i(t+1) \in \arg\max \big[ V(t) + E_V(t) \big] \text{ and } j(t+1) \in \arg\min \big [U(t) + E_U(t) ].
\end{align*}
\end{defn}
A special case of a PFP-system is what Robinson calls a ``vector system,'' which describes fictitious play. This is obtained by setting all entries of $E_V(t)$ and $E_U(t)$ to zero at all timesteps.

\begin{lemma} If $(U,V)$ is an iterative play system for $A$, then 
\begin{align*}
\underset{t \rightarrow \infty}{\lim\inf} \; t^{-1} \bigl\{ \max V(t) - \min U(t) \bigr\} \geq 0.
\end{align*} \label{lemma1}
\end{lemma} 
This lemma follows from the minimax nature of two-player zero-sum games and holds regardless of what rows and columns of $A$ are used to update elements of $U$ and $V$.

\begin{defn} Given an iterative play system $(U,V)$ for matrix $A$, we say row $A_{i,*}$ is {\em $E$-eligible} in the interval $[t, t']$ if for some $t_1 \in [t,t']$, $i$ could have been played as part of AFP. Precisely, the condition is that there exists $t_1 \in [t,t']$ such that 
\begin{align*}
     i \in \arg\max \big[ V(t_1) + E \big] 
\end{align*}
for some $E \in \R^m$ with $\|E\|_\infty < a$. Or, equivalently,
\begin{align*}
    v_i(t_1) \geq \max V(t_1) - 2a.
\end{align*}
Similarly, we say column $j$ is $E$-eligible if
\begin{align*}
   u_j(t_1) \leq \min U(t_1) + 2a. 
\end{align*}

\end{defn}

\begin{lemma} \label{lemma2} If $(U,V)$ is an iterative play system for $A$ and all rows and columns are $E$-eligible in $[s,s+t]$, then we have
\begin{align*}
    &\max U(s+t) - \min U(s+t) \leq 2a(t+1), \text{ and} \\
    &\max V(s+t) - \min V(s+t) \leq 2a(t+1).
\end{align*}
\end{lemma}
\begin{proof}
Let $j \in \arg\max \, U(s+t)$. By the $E$-eligibility of $j$, there must exist $t' \in [s,s+t]$ such that
\begin{align*}
    u_j(t') - \min U(t') \leq 2a.
\end{align*}
So, because the $j$th entry can't change by more than $a$ per timestep,
\begin{align*}
    \max \, U(s+t) &= u_j(s+t) \\ 
    &\leq u_j(t) + at  \\
    &\leq \min \, U(t') + 2a + at  \\
    &\leq \min \, U(t+s) + at + 2a + at,
\end{align*}
where the last inequality holds because for $t' \in [s,s+t]$, the minimum of $U(t')$ versus $U(s+t)$ can't change by more than $at$ in $t$ timesteps. A similar argument shows the result for $V$.
\end{proof}

\begin{lemma} \label{lemma3}
If $(U,V)$ is an iterative play system for $A$ and all rows and columns are $E$-eligible in $[s,s+t]$, then 
\begin{align*}
    \max V(s+t) - \min U(s+t) \leq 4a(t+1).
\end{align*}
\end{lemma}

\begin{proof}
As shown in \citet{robinson1951iterative}, this follows immediately from Lemma \ref{lemma2}. The only difference here is that we replace $2at$ with $2a(t+1)$.
\end{proof}

\begin{lemma} \label{lemma4}
For every matrix $A$, $\epsilon > 0$, there exists a $t_0$ such that for any anticipatory fictitious play system,
\begin{align*}
    \max V(t) - \min U(t) \leq \epsilon t \text{ for all } t \geq t_0.
\end{align*}
\end{lemma}
\begin{proof}
We follow the proof of \citet{robinson1951iterative}, replacing the notion of eligibility with $E$-eligibility. The strategy is induction. If $A \in \R^{1 \times 1}$, the result is trivial because $V(t) = U(t)$ for all $t$. Now assume that the property holds for an arbitrary submatrix $A'$ obtained by deleting any number of columns or rows from $A$. We wish to show that the property also holds for $A$.

Using the inductive hypothesis, pick $t^*$ such that
\begin{align*}
    \max V'(t) - \min U'(t) < \tfrac{1}{2} \epsilon t \text{ for all } t \geq t^*
\end{align*}
for any $A'$ a submatrix of $A$ and $(U',V')$ an anticipatory fictitious play system for $A'$. 

As in \citet{robinson1951iterative}, we wish to show that if in $(U,V)$, some row or column is not $E$-eligible in $[s,s+t^*]$, then 
\begin{align}
    \max V(s+t^*) - \min U(s+t^*) < \max V(s) - \min U(s) + \tfrac{1}{2} \epsilon t^* \label{eqn:robinson3}
\end{align}
Suppose without loss of generality that row $A_{m,*}$ is not $E$-eligible $[s,s+t^*]$. Then we construct a new anticipatory fictitious play system $(U',V')$ for matrix $A'$, which is $A$ with row $m$ deleted. 
Define
\begin{align*}
    U'(t) &= U(s+t) + c \, \ind_n, \\
    V'(t) &= \Proj_m V(s+t),
\end{align*}
for $t=0,1,\dots,t^*$, where $\ind_n$ is a vector of $1$'s with $n$ entries,  $c = \max V(s) - \min U(s)$, and $\Proj_k : \R^m \rightarrow \R^{m-1}$ is the operator that removes entry $k$. We now check the conditions for an anticipatory fictitious play system:
\begin{itemize}
    \item We have $\min U'(0) = \min \{U(s) + [\max V(s) - \min U(s)] \ind_n\} = \max V(s) = \max V'(0)$, where the last equality holds because $m$ is not $E$-eligible, so it could not be a maximizer of $V(s)$, so deleting it to form $V'(0)$ does not change the maximum.
    \item We have that for each  $t$,
    \begin{align*}
        U'(t+1) &= U(s+t+1) + c \ind_n = U(s+t) + A_{i(s+t),*} + c \ind_n = U'(t) + A'_{i(s+t),*}, \\ 
        V'(t+1) &= \Proj_m V(s+t+1) = \Proj_m [V(s+t) + A_{*,j(s+t)}] = V'(t) + A'_{*,j(s+t)},
    \end{align*}
    where $A_{i(s+t),*} = A'_{i(s+t),*}$ because the AFP-ineligibility of $m$ implies $i(s+t) \neq m$.
    \item Finally, we must show that the rows and columns selected still qualify as ``anticipatory'' responses within the context of $(U',V')$ and $A'$, i.e. that
    \begin{align*}
        v'_{i(s+t)}(t) \geq \max V'(t) - 2a \text{ and } u'_{j(s+t)}(t) \leq \min U'(t) + 2a
    \end{align*}
    for each $t=0,1,\dots, t^*$. By the definition of $V'$ and fact that $i(s+t) \neq m$, we have 
    \begin{align*}
        v'_{i(s+t)}(t) &= v_{i(s+t)}(s+t) \\
        &\geq \max V(s+t) - 2a && \text{($(U,V)$ is a PFP-system)}\\
        &= \max V'(s) - 2a,  
    \end{align*}
    and $U'(t)$ is just a shifted version of $U(s+t)$, so the fact that $u_{j(s+t)}(s+t) \leq \min U(s+t) +2a$ implies the result.
\end{itemize}
These points verify that $(U',V')$ satisfy the conditions for an anticipatory fictitious play system for $A'$ on $t=0,\dots,t^*$. We can choose remaining values for both sequences for $t=t^*+1,t^*+2,\dots$ to satisfy the anticipatory fictitious play conditions. So, using the inductive hypothesis, we have
\begin{align*}
    \max V(s+t^*) - \min U(s+t^*) &= \max V'(t^*) - \min \big\{U'(t^*) - [\max V(s) - \min U(s)] \ind_n \big\} \\
    &= \max V'(t^*) - \min U'(t^*) + \max V(s) - \min U(s)  \\
    &<  \tfrac{1}{2} \epsilon t^* + \max V(s) - \min U(s).
\end{align*}

With \eqref{eqn:robinson3} established, we are ready to finish the proof: under the inductive hypothesis, we will be able to deal with time intervals by splitting them into two cases: if a row or column is not $E$-eligible, we apply \eqref{eqn:robinson3}; if all rows or columns are $E$-eligible, we apply Lemma \ref{lemma3}.

Specifically, we show that for any AFP system $(U,V)$ for $A$ and $t \geq 8at^*/\epsilon$,
\begin{align*}
    \max V(t) - \min U(t) < \epsilon t.
\end{align*}
Let $t > t^*$ and express it as $t = (\theta + q) t^*$, where $q \in \N$ and $\theta \in [0,1)$. We consider the collection of length-$t^*$ intervals $[(\theta+r-1)t^*, (\theta+r)t^*)]$ for $r=1,\dots,q$.
\begin{itemize}
    \item {\bf Case 1.} There is at least one interval where all rows and columns are $E$-eligible. Let $[(\theta+s-1)t^*, (\theta+s)t^*)]$ be the latest such interval. By Lemma \ref{lemma3}, 
    \begin{align*}
        \max V[(\theta+s) t^*] - \min V[(\theta+s) t^*] \leq 4a(t^*+1) \leq 8at^*,
    \end{align*}
    where the last inequality holds because $t^* \geq 1$. By choice of the interval, all subsequent intervals with $r=s+1,\dots,q$ have no $E$-eligible rows or columns, so \eqref{eqn:robinson3} gives that
    \begin{align*}
        \max V(t) - \min U(t) \leq \max V[(\theta+s) t^*] - \min V[(\theta+s) t^*] + \tfrac{1}{2} \epsilon t^* (q-s),
    \end{align*}
    noting that the result holds trivially if $q=s$. Combining the previous two results and loosening the bound, we have
    \begin{align}
        \max V(t) - \min U(t) \leq 8at^* + \tfrac{1}{2} \epsilon t^*(q-s) \leq (8a + \tfrac{1}{2} \epsilon q) t^*. \label{eqn:lemma4_case1}
    \end{align}
    
    \item {\bf Case 2.} In each interval $[(\theta+r-1)t^*, (\theta+r)t^*)]$ for $r=1,\dots,q$, some row or column is not $E$-eligible. Applying \eqref{eqn:robinson3} repeatedly,
    \begin{align}
        \max V(t) - \min U(t) &= \max V[(\theta+q) t^*] - \min U[(\theta+q) t^*] \nonumber \\
        &< \max V[(\theta+q-1) t^*] - \min U[(\theta+q-1) t^*] + \tfrac{1}{2} \epsilon t^* \nonumber \\
        &< \max V[(\theta+q-2) t^*] - \min U[(\theta+q-2) t^*] + \tfrac{1}{2} \epsilon t^* + \tfrac{1}{2} \epsilon t^* \nonumber  \\
        &< \dots \nonumber \\
        &< \max V(\theta t^*) - \min U(\theta t^*) + \tfrac{1}{2} q \epsilon t^* \nonumber \\
        &\leq 2a\theta t^* + \tfrac{1}{2} q \epsilon t^* \\
        & = (2a\theta + \tfrac{1}{2} q \epsilon) t^*, \label{eqn:lemma4_case2}
    \end{align}
    where the last inequality holds because $\max V(\theta t^*) \leq a \theta t^*$ and $\min U(\theta t^*) \geq a \theta t^*$.
\end{itemize}
So, comparing \eqref{eqn:lemma4_case1} and \eqref{eqn:lemma4_case2} and noting that $\theta \in [0,1)$, in either case we have that 
\begin{align*}
    \max V(t) - \min U(t) \leq 8at^* + \tfrac{1}{2} \epsilon (q t^*) \leq 8at^* + \tfrac{1}{2} \epsilon t \leq \epsilon t
\end{align*}
for all $t \geq 16 a t^* / \epsilon$.\end{proof}

Finally, we are ready for the proof of Proposition \ref{prop:AFP_converges}, which is essentially identical to the final proof of Theorem 1 in \citet{robinson1951iterative}.
\begin{proof} 
Let $V(0)=0 \in \R^m$, $U(0)=0 \in \R^n$, and $V(t) = t A\ybar_t$, $U(t) = t A^\T \xbar_t$ for $t \in \N$, where $\xbar_t$ and $\ybar_t$ are as given in \eqref{eqn:afp_update}. Clearly, $(U,V)$ forms an iterative play system. It follows from \eqref{eqn:afp_update} that $(U,V)$ is also a PFP-system with $E_V(t) = A \cdot \br^2_A(\xbar_t)$ and $E_U(t) = A^\T \cdot \br^1_A(\ybar_t)$. This is because
\begin{align*}
     \ybar'_t &= \tfrac{t-1}{t} \ybar_{t-1} + \tfrac{1}{t} \br^2_A(\xbar_{t-1}) \text{ implies } \\
     t A \ybar'_t &= (t-1) A \ybar_{t-1} + A\cdot \br^2_A(\xbar_{t-1}) \\
       &= V(t-1) + E_V(t-1),
\end{align*}
so $x_t = \br^1_A(\ybar'_{t-1}) = e_{i(t)}$, where $i(t) \in \arg\max[V(t-1)+E_V(t-1)]$. A similar argument holds for $y_t$.

So, by Lemmas \ref{lemma1} and \ref{lemma4},
\begin{align*}
    \lim_{t \rightarrow \infty} \big( \! \max A \ybar_t - \min \xbar_t^\T A \big) = \lim_{t \rightarrow \infty} \frac{\max V(t) - \min U(t)}{t} = 0,
\end{align*}
where the first equality follows from the definition of $V(t)$ and $U(t)$. Combining this with the fact that, for all $t$, 
\begin{align*}
    \max A \ybar_t &\geq \inf_{y \in \Delta^n} (\max Ay) = v^*, \text{ and} \\
    \min \xbar_t^\T A &\leq \sup_{x \in \Delta^m} (\min x^\T A) = v^*, 
\end{align*}
we have that 
\begin{align*}
    \lim_{t \rightarrow \infty} \max A \ybar_t = \lim_{t \rightarrow \infty} \min \xbar_t^\T A = v^*,
\end{align*}
concluding the proof of convergence of AFP.
\end{proof}

\subsection{Convergence rate of perturbed fictitious play} \label{appendix:fp_worst_case_convergence_rate}
Given a 2p0s matrix game with payoff matrix $A \in \R^{m \times n}$, write $t^*(\epsilon;m,n)$ to denote the value of $t^*$ given by Lemma \ref{lemma4} such that 
\begin{align}
    \max V(t) - \min U(t) < \tfrac{1}{2}\epsilon t \text{ for } t \geq t^*. \label{eqn:conv_rate_proof}
\end{align}

We have that $t^*(\epsilon; 1,1) =1$. So, by the inductive step of the proof of Lemma \ref{lemma4}, $t^*(\epsilon; 2,1) = t^*(\epsilon; 1,2) = \tfrac{8a}{\epsilon}$, which then implies that $t^*(\epsilon;3,1) = t^*(\epsilon;2,2) = t^*(\epsilon;1,3) = (\tfrac{16a}{\epsilon})^2$. Continuing inductively, we see that $t^*(\epsilon;m,n) = (\tfrac{16a}{\epsilon})^{m+n-2}$. Substituting into \eqref{eqn:conv_rate_proof} and rearranging terms gives that
\begin{align*}
    \frac{\max V(t) - \min U(t)}{t} < \epsilon \text{ for } t \geq (\tfrac{8a}{\epsilon})^{m+n-2}.
\end{align*}
Choosing $\epsilon_t = \frac{8a}{t^{1/(m+n-2)}}$ for each $t$ gives the result.


\section[Proofs for transitive and cyclic games]{Proof of Proposition \ref{prop:comparison_on_transitive_and_cyclic}}
\label{appendix:convergence_rate_proofs}

For first two subsections, we restate the definitions of $\Delta_t$, $t \in \N_0$ from \eqref{eqn:delta_sequence_update}: $\Delta_0 = [0, \dots, 0]^\T \in \mathbb{Z}^n$, $\Delta_t = t C^n \, \xbar_t$ for each $t \in \N$, and $i_t$ is the index played by FP (AFP) at time $t$, so
    \begin{align*}
        \Delta_{t+1,j} = \begin{cases}
        \Delta_{t,j}-1 &\text{if } j=i_t-1 \mod n; \\ 
        \Delta_{t,j}+1 &\text{if } j=i_t+1 \mod n; \hspace{1cm} \eqref{eqn:delta_sequence_update}\\
        \Delta_{t,j}  &\text{otherwise;} 
        \end{cases} 
    \end{align*}
for each $t \in \N_0$ and $j \in \{1,\dots, n\}$. 

\subsection[FP: cyclic game]{FP: $C^n$}
We must show that $\max \Delta_t = \Omega_p(\sqrt{t})$ under random tiebreaking. Throughout, whenever performing arithmetic with indices, that arithetic is done modulo $n$. As in the body of the paper, define $t_m = \inf \{ t \in \N_0 : \max \Delta_t = m\}$ and note the Markov inequality bound:
\begin{align*}
    P( \max \Delta_t < m) = P(t_m > t) \leq E(t_m)/t= \frac{1}{t} \sum_{k=0}^{m-1} E(t_{k+1} - t_{k}).
\end{align*}
The bulk of the argument is in finding a bound for $E(t_{k+1} - t_{k})$.

It follows by the definition of $\Delta_t$ and the FP update that $i_t \in \arg\max \Delta_t$. Index $i_1$ may be chosen arbitrarily, but for $t > 1$, it follows from \eqref{eqn:delta_sequence_update} that
\begin{align*}
    i_{(t+1)} \in \begin{cases} 
    \{ i_t \} &\text{if } \Delta_{t,i_t} > \Delta_{t,i_t+1};\\
    \{ i_t, i_t+1 \} &\text{if } \Delta_{t,i_t} = \Delta_{t,i_t+1}; \\
    \{ i_t+1 \} &\text{if } \Delta_{t,i_t} <\Delta_{t,i_t+1};
    \end{cases}
\end{align*}
because the value that is incremented at time $t$ is the value adjacent to index $i_t$, $\Delta_{t,i_t+1}$. Let $\tau_1 = 1$ and inductively define, for $\ell \in \N$, $\tau_{\ell+1} = \inf \{ t > \tau_\ell : \Delta_{t,i_t} = \Delta_{t,i_t+1} \}$ to be the next time at which there are two possible choices for $i_{t+1}$. This is depicted in Table \ref{tab:fp_rps_delta_update}, writing $m = \max \Delta_{\tau_\ell}$ and $m' \in \{m, m+1\}$.

\begin{table}[h]
    \centering
    \caption{The process of incrementing the index played under FP on $C^n$.}
    \begin{tabular}{ccccccccc}
    $\Delta_{\tau_\ell} = $ & $[\dots$ & $\leq0$ &  $m$  & $m$ & $\leq$ 0 & $\leq 0 $ & $\dots] $\\
    & &  \color{lightgray}$\hphantom{-1}\Big\downarrow {-1}$    & $\hphantom{-a} \Big\downarrow {-a}$ &  \color{lightgray}$\hphantom{+1}\Big\downarrow {+1}$   & $\hphantom{+a} \Big\downarrow {+a}$  &  \\
    $\Delta_{\tau_{(\ell+1)}} = $   & $[\dots$    &  $\leq0$     & $m-a$  & $m'$ & $m'$ & $\leq 0 $ & $\dots] $\\
    \end{tabular}
    \label{tab:fp_rps_delta_update}
\end{table}

As shown in the table, all entries of $\Delta_t$ other than the two maximum values must be nonpositive at each $t = \tau_\ell$. This follows by induction, since it holds for $\tau=1$ ($\Delta_{\tau_1}$ has one positive entry) and if it's true for some $\ell \in \N$, then in order to progress to $\tau_{(\ell+1)}$, we must add some number $a > m$ (over the course of $a$ timesteps) to the next entry, which means by \eqref{eqn:delta_sequence_update} we will subtract $a$ from the previous entry $m$, with $m-a \leq 0$. Finally, note that between $\tau_\ell$ and $\tau_{(\ell+1)}$, either we will have incremented the max from $m$ to $m+1$ if $i_{\tau_\ell}=i_{(\tau_\ell-1)}$, or we will not have not (if $i_{\tau_\ell}=i_{(\tau_\ell-1)}+1)$, and the max will remain at $m$ until we repeat some further number of increments of the index played. It is only on these timesteps that the maximum value can increment.

Based on this reasoning, we know that for any $k$, there must exist $\ell(k)$ such that $t_{k} = \tau_{\ell(k)}+1$. Consider the random variable $\ell(k+1) - \ell(k)$, which is the number of increments of the index played that occurred between the increment of the max from $k$ to $k+1$. Under uniform random tiebreaking, we have that $\ell(k+1) - \ell(k) \sim \text{Geometric}(1/2)$, since at each $\tau_\ell$ there is a 1/2 chance of incrementing (``success'') or not incrementing (``failure''). So, $E[\ell(k+1) - \ell(k)] = 2$.
Now suppose that we had the bound $\tau_{\ell+1} - \tau_\ell \leq d \max \Delta_{\tau_\ell}$ for some $d>0$. That would imply that 
\begin{align*}
    t_{k+1} - t_{k} = \tau_{\ell(k+1)} - \tau_{\ell(k)} = \sum_{r=\ell(k)}^{\ell(k+1)-1} \tau_{r+1} - \tau_r &\leq \sum_{r=\ell(k)}^{\ell(k+1)-1} d \max \Delta_{t_{k+1}} \\  &= d [\ell(k+1) - \ell(k)] (k+1).
\end{align*}
Taking the expectation of both sides, we get $E(t_{k+1}-t_k) \leq 2d(k+1)$. Plugging this into the Markov bound is sufficient to finish the argument, as explained in the proof sketch for Proposition \ref{prop:comparison_on_transitive_and_cyclic} (in the paper).

All that remains is to show that $\tau_{\ell+1} - \tau_\ell \leq d \max \Delta_{\tau_\ell} = dm$. From the argument depicted in Table \ref{tab:fp_rps_delta_update}, we know $\tau_{\ell+1} - \tau_\ell \leq a+1$, and that $a \leq m+1 - \min \Delta_{\tau_\ell}$. Because $\Delta_{\tau_\ell}$ has only two positive entries and $\sum_{i=1}^n \Delta_{t,i} = 0$, we have $\min \Delta_{\tau_\ell} \geq -2m$, so $\tau_{\ell+1} - \tau_\ell \leq 3m +2 \leq 5m$, concluding the proof.

\subsection[AFP: cyclic game]{AFP: $C^n$, $n=3,4$}
For $n=3,4$ we show $\max_t \Delta_t < 3$, which proves the result.

Based on the AFP update, it is impossible to get to $\max_t \Delta_t = m+1$ unless there are at least two non-adjacent $m$'s in $\Delta_{t-1}$ with an $m-1$ in between. Otherwise, the two-step nature of the AFP update will not allow an $m$ to be incremented to $m+1$. However, it is impossible to have two non-adjacent $m$'s with an $m-1$ in between for $n=3, m=2$ because the entries of $\Delta_{t-1}$ sum to 0. Furthermore, in the $n=4$ case, for each $t$, it must be that $\Delta_t = [a,b,-a,-b]$ for some $a$ and $b$, by \eqref{eqn:delta_sequence_update}. So there also cannot be three positive numbers in this case.

\subsection[FP: transitive game]{FP: $T^n$}
Assume without loss of generality that $x_1 = e_1$. Let $\tau_k= \min \{t : x_t = e_k\}$ and note that the form of $T^n$ implies that the strategies played by FP must be nondecreasing and increment by at most 1 at a time. We argue by strong induction that $\tau_{k+1} - \tau_k > \tau_{k} - \tau_{k-1}$ for each $k <n$. Checking the first few terms, we have
\begin{align*}
    \tau_1 = 1, \text{ so } &T^n \xbar_1 = n^{-1} [0, n, 0, \dots,0]^\T, \text{ so} \\
    \tau_2 = 2, \text{ so } &T^n \xbar_2 = 2^{-1} n^{-1} [-n, n, n-1, \dots,0]^\T, \text{ so}
\end{align*}
$x_3 = e_2$, and therefore $\tau_3 > 3$, so $\tau_3 - \tau_2 > \tau_2 - \tau_1 = 1$. Now assume that for some fixed $k <n$ and all $k' \in\{1,\dots,k\}$ that $\tau_{k'+1} - \tau_{k'} > \tau_{k'} - \tau_{k'-1}$. Note that 
\begin{align*}
    \xbar_{\tau_{(k+1)}-1} &= [ \tau_2 -\tau_1, \dots, \tau_{k+1} - \tau_k, 0, 0, \dots, 0]^\T, \text{ so,}\\ 
    T^n \xbar_{\tau_{(k+1)}-1} &\propto  [<, <, \dots, <, (\tau_{k+1} - \tau_{k})(n-k+1), 0, \dots, 0]^\T, \text{ and} \\
    T^n \xbar_{\ell} &\propto  [<, <, \dots, <, (\tau_{k+1} - \tau_{k}) (n-k+1), (\ell-\tau_{k+1}+1) (n-k), \dots, 0]^\T,
\end{align*}
for $\ell \in \{\tau_{k+1}, \dots,  \tau_{k+2}-1\}$, where the `$<$' signs indicate values that are no greater than their neighbors on the right; this holds by the inductive hypothesis and definition of $T^n$. We know that for steps $\tau_{k+1}, \dots \tau_{k+2}-1$, FP will play $e_{k+1}$, and we know that $\tau_{k+2}$ is the first timestep at which $(\tau_{k+2} - \tau_{k+1})(n-k) \geq (\tau_{k+1} - \tau_{k})(n-k+1)$ (or else FP would have played $k+1$ at $\tau_{k+2}$, a contradiction). It follows that $\tau_{k+2} -\tau_{k+1} > \tau_{k+1}-\tau_k$, as desired.

This result implies that $\tau_{k+1} - \tau_k \geq k$ for each $k <n$, so we have $\tau_k \geq \sum_{j=1}^k j = k(k+1)/2 \geq k^2/2$ for each $k$. Inverting this, we get that $t \mapsto \sqrt{2t}$ is an upper bound on $k(t)$, the index played by FP at time $t$. Combining the expression for $T^n \xbar_\ell$, $\ell \in \{\tau_{k+1}, \dots,  \tau_{k+2}-1\}$, with this bound, we get that $\max \, T^n \xbar_t = n^{-1} t^{-1} (\tau_{k(t)+1} - \tau_{k(t)})(n-k(t)+1) \geq n^{-1} t^{-1} (n - k(t)+1) \geq n^{-1} t^{-1} (n-\sqrt{2t}) = \Omega(1/\sqrt{t})$.

\subsection[AFP: transitive game]{AFP: $T^n$}
We argue first by strong induction that $x_t = e_{\min(t, n)}$ for all $t$. Assume without loss of generality that $x_1 = e_1$. Now assume that, for some fixed $\tau$, $x_t = e_{\min(t,n)}$ for $t \leq \tau$. 

If $\tau < n$, then under the inductive hypothesis, 
\begin{align*}
    T^n \xbar_\tau &= \tau^{-1} T^n \; [\overbrace{1, 1, \dots, 1}^{\tau}, \; 0, \dots, 0]^\T \\
    &= \tau^{-1} n^{-1} [-n, 1, \dots, (n-\tau+2), \; (n-\tau+1), \dots, 0]^\T,
\end{align*}
for which the largest value is at index $\tau$, so $x' = e_\tau$, so then $T^n \xbar'_\tau$ will have largest value $2n^{-1} (n-\tau+1)$ at index $\tau+1$, so $x_{\tau+1} = e_{\tau+1}$. If $\tau \geq n$, then under the inductive hypothesis, $T^n \xbar_\tau = \tau^{-1} T^n [1, 1, \dots, 1, \tau-n+1]^\T = \tau^{-1} n^{-1} [-n, 1, 1, \dots, 1, 1-2(\tau-n), 2]^\T$, at which point $x'_\tau = x_{\tau+1} = e_n$, as desired. 

Finally, we are interested in $\max T^n \xbar_t$ for $t < n$, which we obtain from the calculation above, $\frac{n-t+2}{nt} = O(1/t)$.

\clearpage
\section{Additional figures}
\label{appendix:additional_figures}

\begin{figure}[h]
    \centering
    \includegraphics[width=0.47\textwidth]{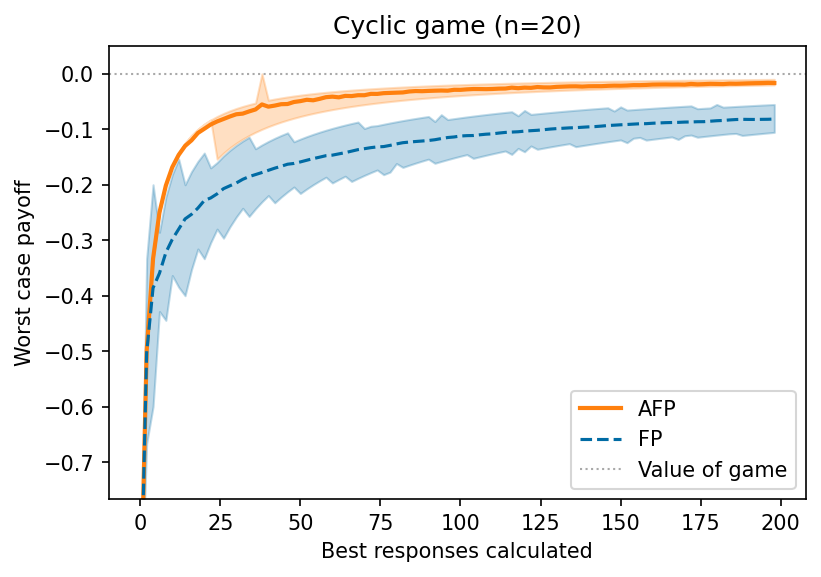}
    \includegraphics[width=0.47\textwidth]{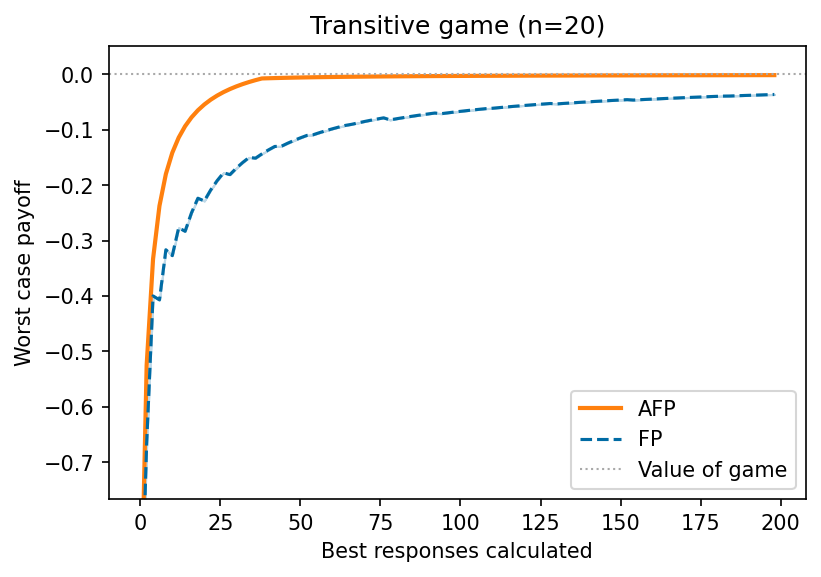}
    \caption{Comparisons of FP and AFP on $C^{20}$ and $T^{20}$ with random tiebreaking. As before, highlighted regions indicate 10th and 90th percentiles across 10,000 runs.}
\end{figure}

\begin{figure}[h]
\centering
    \includegraphics[width=0.47\textwidth]{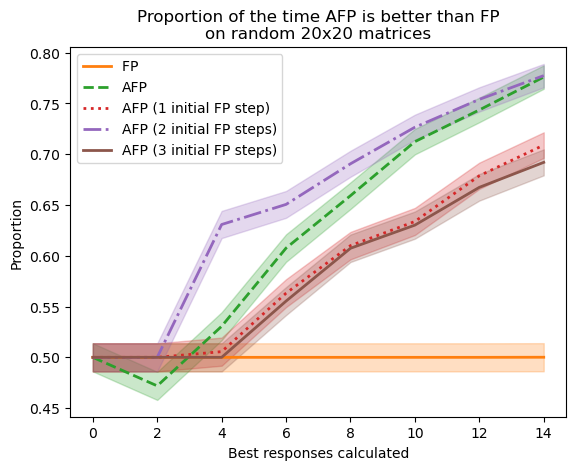}
    \caption{Comparison of FP with versions of AFP initialized with different numbers of steps of FP, based on worst case performance, as in Figure \ref{fig:proportion_afp_better_fp}. When there is an exact tie, credit is split evenly, resulting in a solid line at 0.5 for FP compared with itself.}
\end{figure}

\clearpage
\section{RL experiment hyperparameters} \label{appendix:rl_experiment_hyperparameters}

    \begin{table}[h]
    \begin{tabular}{ccc}
    \toprule
   {\bf APPO} & TinyFighter & RWS  \\ \midrule
       Discount factor $\gamma$ & 0.99  & 0.995 \\
      Value function coefficient &  0.5 & 0.5 \\ 
      Entropy coefficient & 0.01 & 0.02 \\
      Learning rate & 3e-4 & 1.5e-4 \\ 
      Gradient clipping & 10  & 80 \\
      Batch size & 5,120 & 6,000 \\
      Number of workers & 80 & 319 \\
      Rollout fragment length & 64 & 64 \\
      Other &  \multicolumn{2}{c}{\em Ray defaults (as of Nov 2022) \citep{moritz2018ray}} 
       \\
      \bottomrule
    \end{tabular}
    \end{table}
    
    \begin{table}[h]
    \begin{tabular}{cp{12cm}}
    \toprule
    \multicolumn{2}{c}{\bf Running With Scissors Environment} \\ \midrule
    Reward shaping & Policy $t=1$ played against random and received a unit of reward for each `scissors' picked up and a negative unit of reward for each `rock' and `paper' picked up. \\
      Reward scaling ($t=1$) & 10 \\
      Reward scaling ($t>1$) & 100 \\
      Sprite size & 3x3 pixels \\
      Agent field of view & 5x5 grid units (15x15 pixels)   \\
      Agent memory & 4 frames \\
      \bottomrule
    \end{tabular} 
    \end{table}

\begin{table}[ht]
    \caption{Population neural network $\Pi_\theta(t)$ architecture for TinyFighter. The neural network takes as inputs: an image array, the ratios of rock, paper, and scissors in the player's inventory, and an agent index $t$ which gets mapped to a vector in $R^{128}$ by an embedding layer. 
    The Shape column does not include batch dimension. We use ReLU activation functions except for the output layers, which use the identity function (`linear'). We used the Keras library \citep{chollet2015keras} to implement the model.} 
    \scriptsize
    \centering
\begin{tabular}{lllll}
\toprule
                Name &         Type &  Shape & \# Parameters &                                              Input \\
\midrule
               state &   InputLayer &   [13] &            0 &                                                    \\
               dense &        Dense &  [128] &        1,792 &                                              state \\
             dense\_1 &        Dense &  [128] &       16,512 &                                              dense \\
 neupl\_opponent\_dist &   InputLayer &    [4] &            0 &                                                    \\
  policy\_head\_inputs &  Concatenate &  [145] &            0 &                state, dense\_1, neupl\_opponent\_dist \\
  neupl\_opponent\_idx &   InputLayer &    [5] &            0 &                                                    \\
             dense\_2 &        Dense &  [256] &       37,376 &                                 policy\_head\_inputs \\
   value\_head\_inputs &  Concatenate &  [150] &            0 &  state, dense\_1, neupl\_opponent\_dist, neupl\_opp... \\
             dense\_3 &        Dense &  [128] &       32,896 &                                            dense\_2 \\
             dense\_6 &        Dense &  [256] &       38,656 &                                  value\_head\_inputs \\
             dense\_4 &        Dense &   [64] &        8,256 &                                            dense\_3 \\
         action\_mask &   InputLayer &    [4] &            0 &                                                    \\
             dense\_7 &        Dense &  [128] &       32,896 &                                            dense\_6 \\
             dense\_5 &        Dense &    [4] &          260 &                                            dense\_4 \\
              lambda &       Lambda &    [4] &            0 &                                        action\_mask \\
             dense\_8 &        Dense &   [64] &        8,256 &                                            dense\_7 \\
          policy\_out &          Add &    [4] &            0 &                                    dense\_5, lambda \\
           value\_out &        Dense &    [1] &           65 &                                            dense\_8 \\
\bottomrule
\end{tabular}

\end{table}

    \begin{table}[ht]
    \caption{Population neural network $\Pi_\theta(t)$ architecture for RWS. The neural network takes as inputs: an image array, the ratios of rock, paper, and scissors in the player's inventory, and an agent index $t$ which gets mapped to a vector in $R^{128}$ by an embedding layer. 
    The Shape column does not include batch dimension. We use ReLU activation functions except for the output layers, which use the identity function (`linear'). We used the Keras library \citep{chollet2015keras} to implement the model.} 
    \scriptsize
    \centering
\begin{tabular}{lllll}
\toprule
                Name &         Type &           Shape & \# Parameters &                                          Input \\
\midrule
         rgb\_history &   InputLayer &  [4, 15, 15, 3] &            0 &                                                \\
           rescaling &    Rescaling &  [4, 15, 15, 3] &            0 &                                    rgb\_history \\
              conv2d &       Conv2D &   [4, 5, 5, 32] &          896 &                                      rescaling \\
          activation &   Activation &   [4, 5, 5, 32] &            0 &                                         conv2d \\
            conv2d\_1 &       Conv2D &   [4, 5, 5, 16] &          528 &                                     activation \\
        activation\_1 &   Activation &   [4, 5, 5, 16] &            0 &                                       conv2d\_1 \\
             reshape &      Reshape &        [4, 400] &            0 &                                   activation\_1 \\
  inv\_ratios\_history &   InputLayer &          [4, 3] &            0 &                                                \\
         concatenate &  Concatenate &        [4, 403] &            0 &                    reshape, inv\_ratios\_history \\
               dense &        Dense &        [4, 256] &      103,424 &                                    concatenate \\
             dense\_1 &        Dense &        [4, 256] &       65,792 &                                          dense \\
                lstm &         LSTM &           [256] &      525,312 &                                        dense\_1 \\
 neupl\_opponent\_dist &   InputLayer &            [11] &            0 &                                                \\
  neupl\_opponent\_idx &   InputLayer &            [12] &            0 &                                                \\
  policy\_head\_inputs &  Concatenate &           [267] &            0 &                      lstm, neupl\_opponent\_dist \\
   value\_head\_inputs &  Concatenate &           [279] &            0 &  lstm, neupl\_opponent\_dist, neupl\_opponent\_idx \\
            policy\_1 &        Dense &           [512] &      137,216 &                             policy\_head\_inputs \\
             value\_1 &        Dense &           [512] &      143,360 &                              value\_head\_inputs \\
            policy\_2 &        Dense &           [256] &      131,328 &                                       policy\_1 \\
             value\_2 &        Dense &           [256] &      131,328 &                                        value\_1 \\
            policy\_3 &        Dense &           [128] &       32,896 &                                       policy\_2 \\
             value\_3 &        Dense &           [128] &       32,896 &                                        value\_2 \\
          policy\_out &        Dense &            [30] &        3,870 &                                       policy\_3 \\
           value\_out &        Dense &             [1] &          129 &                                        value\_3 \\
\bottomrule
\end{tabular}
\end{table}

\clearpage
\section{Vanilla RL implementations of FP and AFP} \label{appendix:vanilla_rl_fp_afp}
Algorithm \ref{alg:fp_afp_rl} gives a simple reinforcement learning implementation of FP and AFP that does not rely on Neural Population Learning or the reservoir buffer sampling of \citet{heinrich2015fictitious}. Notes:
\begin{itemize}
    \item $-i$ refers to the opponent of $i$.
    \item $\textsc{ReinforcementLearning}(\pi, \mu)$ plays $\pi$ against $\mu$ for some number of episodes, gathers data, and updates $\pi$ by reinforcement learning.
    \item Lines \ref{alg:sample} and \ref{alg:train} are the approximate reinforcement learning analogue to computing a best response to the average of all previous strategies. For details, see \citet{heinrich2015fictitious}, which uses a more complicated setup in order to limit storage requirements (constantly learning $\beta$ by supervised learning).
\end{itemize}

\begin{algorithm}[ht]
\caption{FP/AFP with reinforcement learning}
\label{alg:fp_afp_rl}
\begin{algorithmic}[1]
\State Choose setting: \texttt{FP} or \texttt{AFP}.
\State Initialize policies $\pi^1_1$ and $\pi^2_1$.
\State Initialize policy stores $\Pi^1 = \{\pi^1_1\}, \Pi^2 = \{\pi^2_1\}$.
\For{$t=1,2,\dots$}
    \State Initialize $\pi^1_{t+1}, \pi^2_{t+1}$.
    \For{$i=1,2$} 
       \While{per-timestep compute budget remains}
            \State Sample $\beta \sim \text{Uniform}(\Pi^{-i})$. \label{alg:sample}
            \State $\pi^i_{t+1} \gets \textsc{ReinforcementLearning}(\pi^i_{t+1}, \beta)$. \label{alg:train}
        \EndWhile
        \State $\Pi^i \gets \Pi^i \cup \{ \pi^i_{t+1} \}$.
    \EndFor
        \If{setting is \texttt{AFP} {\bf and} $t$ is odd}
        \State $\Pi^1 \gets \Pi^1 \setminus \{ \pi^1_{t} \}$ and $\Pi^2 \gets \Pi^2 \setminus \{ \pi^2_{t} \}$.
    \EndIf
\EndFor
\end{algorithmic}

\end{algorithm}

\end{document}